\def\noeditingmarks{}  %

\documentclass[sigplan]{acmart}
\renewcommand\footnotetextcopyrightpermission[1]{} %
\setcopyright{none}
\AtBeginDocument{\pagestyle{plain}}
\acmDOI{}
\acmISBN{}
\acmPrice{}

\usepackage[utf8]{inputenc}
\usepackage{mdframed}
\usepackage{subcaption} 
\usepackage{booktabs}  %
\usepackage{dcolumn}   %
\usepackage{multirow}  %
\usepackage{rotating}
\usepackage{xspace}
\usepackage{hyphenat}  %

\usepackage{color}
\usepackage{enumerate}
\usepackage{multicol}
\usepackage{fancyvrb}
\fvset{%
fontsize=\footnotesize,
numbers=left,
xleftmargin=18pt
}
\usepackage{textcomp}
\usepackage{amsmath}
\usepackage{amsfonts}
\usepackage{wasysym}
\usepackage{tabularx}
\usepackage{pifont}
\usepackage{microtype} %
\usepackage{mathtools} %
\usepackage{mathrsfs} %

\usepackage{url}

\usepackage{ulem}
\usepackage{array}    %

\usepackage{float}

\usepackage[noend]{algpseudocode}
\algrenewcomment[1]{\hfill// #1}%
\algtext*{Function}
\algrenewcommand{\algorithmicrequire}{\textbf{Input:}}
\algrenewcommand{\algorithmicensure}{\textbf{Output:}}
\algnewcommand{\LineComment}[1]{\State // #1}
\algrenewcommand\textproc{}%

\usepackage[noeka]{mathrmletter}
\usepackage{arydshln}
\usepackage[english]{babel}
\usepackage{blindtext}

\usepackage[font=small]{caption}
\usepackage[labelfont=small]{caption}
\usepackage[labelfont=small]{subcaption}

\usepackage{bigdelim}
\usepackage{siunitx}
\newif\ifintrouble
\newif\ifcuttext

\ifintrouble
\cuttexttrue
\else
\cuttextfalse
\fi

\usepackage{colortbl} %

\newcommand{\textred}[1]{\begingroup \color{red} #1\endgroup}
\ifx\noeditingmarks\undefined
   \newcommand{\pgwrapper}[2]{\textred{#1: #2}}
   \newcommand{\pgwrapperb}[1]{\textbf{#1}}
\else
   \newcommand{\pgwrapperb}[1]{}
   \newcommand{\pgwrapper}[2]{}
\fi

\ifx\noeditingmarks\undefined
    \newcommand{\changebars}[2]{%
    [{\em \begingroup \color{magenta} #1 \endgroup}]
    [\begingroup \color{magenta} \sout{#2} \endgroup]}
    
\else
    \newcommand{\changebars}[2]{#1}
    
\fi

\newcommand{\sys}{{Aero}\xspace}
\newcommand{\Sys}{\sys}

\makeatletter
\renewcommand*{\@fnsymbol}[1]{\ensuremath{\ifcase#1\or \star\or \dagger\or \ddagger\or
   \mathsection\or \mathparagraph\or \|\or **\or \dagger\dagger
   \or \ddagger\ddagger \else\@ctrerr\fi}}
\makeatother

\def\hn{\usefont{OT1}{phv}{mc}{n}\selectfont}

\newcommand{\mpfont}{\hn\scriptsize}
\newcommand{\MPworker}[2]{{\color{#1}\vrule\vrule}{\marginpar{\color{#1}\mpfont #2}}}
\ifx\noeditingmarks\undefined
    \newcommand{\MP}[1]{\MPworker{red}{#1}}
    \newcommand{\MPtg}[1]{\MPworker{red}{#1}}
    \newcommand{\MPnc}[1]{\MPworker{blue}{#1}}
    \newcommand{\MPkl}[1]{\MPworker{brown}{#1}}
\else
   \newcommand{\MP}[1]{}
    \newcommand{\MPtg}[1]{}
     \newcommand{\MPnc}[1]{}
    \newcommand{\MPkl}[1]{}
\fi
\newcommand\rmv[1]{}

\newcommand{\techReportOnly}[1]{}

\newcommand{\main}{{\small\textsc{Main}}\xspace}
\newcommand{\userupdate}{{\small\textsc{UserUpdate}}\xspace}

\newcommand{\encplain}{{\small\textsc{Enc}}\xspace}

\DeclarePairedDelimiterX{\norm}[1]{\lVert}{\rVert}{#1}

\usepackage{tkz-euclide}
\usepackage{tikz}

\usepackage{mathtools}

\newcommand{\cpu}{\textsc{cpu}\xspace}
\newcommand{\cpus}{\textsc{cpu}s\xspace}

\def\compactify{\itemsep0in \topsep0.5pt \parsep=0.00in \partopsep=0pt
\leftmargin2em}
\let\latexusecounter=\usecounter

\newenvironment{myenumerate}
  {\def\usecounter{\compactify\latexusecounter}
   \begin{enumerate}}
  {\end{enumerate}\let\usecounter=\latexusecounter}

\newenvironment{myenumerate2}
  {\def\usecounter{\itemsep=0ex \topsep0.7ex \parsep=1ex \partopsep=0pt
    \leftmargin1.5em\latexusecounter}
   \begin{enumerate}}
  {\end{enumerate}\let\usecounter=\latexusecounter}

\newenvironment{myitemize}%
  {\begin{list}{\labelitemi}{\itemsep3pt \topsep3pt \parsep0.00in
  \partopsep=3pt \leftmargin1em}}%
  {\end{list}}

  {\begin{list}{\labelitemi}{\itemsep6pt \topsep6pt \parsep0.00in
  \partopsep=0pt \leftmargin\parindent}}%
  {\end{list}}

  {\begin{list}{\labelitemi}{\itemsep3pt \topsep3pt \parsep0.00in
  \partopsep=0pt \leftmargin\parindent}}%
  {\end{list}}
  {\begin{list}{\labelitemi}{\itemsep0in \topsep2pt \parsep0.00in
  \partopsep=0pt \leftmargin\parindent}}%
  {\end{list}}

  {\begin{list}{\labelitemi}{\itemsep3pt \topsep3pt \parsep0.00in
  \partopsep=0pt \leftmargin1.5em}}%
  {\end{list}}

\def\emparagraph#1{\vspace{1mm}\noindent{\bf #1}}

\usepackage{amsthm}

\def\discretionaryslash{\discretionary{/}{}{/}}
{\catcode`\/\active
\gdef\URLprepare{\catcode`\/\active\let/\discretionaryslash
        \def~{\char`\~}}}%
\def\URL{\bgroup\URLprepare\realURL}%
\def\realURL#1{\tt #1\egroup}%

\begin{document}

\newcommand{\supsyml}[1]{\raisebox{4pt}{\footnotesize #1}}
\newcommand{\rstar}{\supsyml{$\ast$}\xspace}
\newcommand{\rdag}{\supsyml{$\ast\ast$}\xspace}

\date{}

\title{\LARGE{Federated learning with differential privacy and an untrusted aggregator (technical report)}}

\author{Kunlong Liu}
\affiliation{\institution{\fontsize{9.5}{11}\selectfont \textit{University of California, Santa Barbara}}}
\author{Trinabh Gupta}
\affiliation{\institution{\fontsize{9.5}{11}\selectfont \textit{University of California, Santa Barbara}}}

\begin{abstract}
Federated learning for training models over mobile devices is gaining
popularity.
Current systems for this task 
exhibit significant trade-offs between model accuracy, privacy guarantee, and device 
efficiency. 
For instance, Oort (OSDI 2021) provides excellent accuracy and efficiency but requires
a trusted central server. On the other hand, 
 Orchard (OSDI 2020) provides good accuracy and 
the rigorous guarantee of differential privacy 
over an untrusted server, 
but creates huge overhead for the devices.
This paper describes \sys, 
a new federated learning system that significantly improves this
trade-off. 
\sys guarantees good accuracy, differential privacy over an untrusted server,
and keeps the device overhead low.
The key idea of \sys is to tune system architecture and design 
to a specific set of popular, 
federated learning algorithms. 
This tuning requires novel optimizations and techniques, e.g., a new protocol 
to securely aggregate updates from devices.
An evaluation of \sys demonstrates that it 
provides comparable accuracy to plain 
federated learning (without differential privacy), and it
improves efficiency (\cpu and network)
over Orchard by up to $10^{5}\times$.
\end{abstract}

\maketitle
\section{Introduction}
\label{s:intro}
\label{s:introduction}

Federated learning (FL) is a recent paradigm in machine learning that embraces a decentralized training architecture~\cite{mcmahan2017communication}. In contrast to the traditional, central model of learning where users \emph{ship} their training data to a central server, 
users in FL download the latest model parameters from the  server, 
perform \emph{local} training to generate \emph{updates} to the parameters, and send only the updates to the server. Federated learning has gained popularity in training models for mobiles~\cite{hard2018federated,hartmann2019federated, kairouz2021practical,DPFTRL} as it
can save  network bandwidth and it is privacy-friendly---raw data stays at the devices.

Current systems for federated learning exhibit significant trade-offs
between model accuracy, privacy, and device efficiency. 
For instance, one class of systems that includes 
Oort~\cite{lai2021oort}, FedScale~\cite{fedscale-icml22}, and 
FedML~\cite{chaoyanghe2020fedml} provides excellent accuracy (comparable to centralized learning) and device efficiency. But these systems provide only a weak
notion of privacy. This point is subtle. 
At first glimpse, it appears that in any federated learning system, since
users ship updates to model parameters rather than the raw training data, this data 
(user images, text messages, search queries, etc.) remains confidential.
However, research shows that updates can be reverse-engineered to reveal the raw data~\cite{zhu2019deep,melis2019exploiting,shokri2017membership}. 
Thus, if the server is compromised, so is the users' data. In other words, the server must be trusted.

On the other hand, systems such as HybridAlpha~\cite{xu2019hybridalpha}
and Orchard~\cite{roth2020orchard} offer good accuracy 
and a differential privacy guarantee for users' data. Informally, differential privacy
says that an adversary cannot deduce a user's training data by inspecting the updates or the learned model parameters~\cite{dwork2011firm,dwork2006calibrating,dwork2014algorithmic,abadi2016deep}.
In fact, Orchard guarantees differential privacy while assuming a byzantine server. But the downside is \changebars{the}{that 
it creates} high overhead for the devices. 
For example, to train
a CNN model with
1.2 million parameters~\cite{reddi2020adaptive}, Orchard requires from each device $\approx$14 minutes of training time on a six-core processor and $\approx$840~MiB in network transfers\changebars{ per round}{--- per-round}  of training (\S\ref{s:eval:Orchard}). 
The full training requires at least a few hundred rounds. 
Further, for a few randomly chosen devices\changebars{}{ that do additional work}, this per-round cost \changebars{spikes to}{is} 
$\approx$214 hours of \cpu time and $\approx$11~TiB of network transfers.
Clearly, this is quite high.\footnote{Another class of systems provides a particular type 
of differential privacy called local differential privacy (LDP)~\cite{duchi2013local,ijcai2021-217,truex2020ldp}.
These systems are efficient but LDP creates a high accuracy loss~\cite{truex2020ldp,grafberger2021fedless,ijcai2021-217} (\S\ref{s:problem:solutions}, \S\ref{s:relwork}).} %

This paper describes a new federated learning system, \sys, that significantly improves the 
tradeoff between accuracy, privacy, and device overhead.
\Sys provides good accuracy, the differential privacy guarantee in the same threat model as Orchard,
and low device overhead. For instance, most of the time \sys's devices
incur overhead in milliseconds of \cpu time and KiBs of network transfers.

The key idea in \sys is that it does
not aim to be a general-purpose
federated learning system, rather focuses
on a particular class of algorithms (\S\ref{s:dpfedavg}).
These algorithms sample devices that contribute updates 
in a round using a simple probability parameter (e.g., a device is selected
with a probability of $10^{-5}$), then aggregate updates across devices
by averaging them, and generate noise needed for differential privacy
from a Gaussian distribution. 
Admittedly, this is only one class of algorithms, but this class comprises popular algorithms such as 
DP-FedAvg~\cite{mcmahan2017learning} and DP-FedSGD~\cite{mcmahan2017learning} that are the ones
commonly used and deployed~\cite{hard2018federated,hartmann2019federated,fedscale-icml22,chaoyanghe2020fedml}.
With this restriction, \sys tunes
system architecture and design to these algorithms, thereby 
 gaining on performance by orders of magnitude.

This tuning is non-trivial and requires novel techniques and optimizations (\S\ref{s:arch}, \S\ref{s:design}).
As one example, the devices must verify that the byzantine server did not behave maliciously.
One prior technique is to use a summation tree~\cite{roth2019honeycrisp}, where the server explicitly shows its 
work aggregating updates across devices in a tree form; the devices then
collectively check nodes of this tree. This checking, in turn, 
adds overhead to the devices. \Sys addresses this tension between privacy and overhead by 
leveraging the sampling characteristic of DP-FedAvg and similar algorithms: 
the total number of devices that participate in the system (e.g., one billion) 
is much larger than the ones that are sampled to contribute updates. 
Leveraging this characteristic, \sys employs multiple, finer-grained summation trees (rather than a monolithic tree) to 
massively divide the checking work across the large device pool (\S\ref{s:design:add}). 
\Sys further optimizes how each device verifies the nodes of the tree using 
a technique called polynomial identity testing (\S\ref{s:design:add}).
 The aforementioned is just one example of optimization; \sys uses multiple throughout
its architecture (\S\ref{s:arch}) and design (\S\ref{s:design}).

We implemented \sys by extending the FedScale FL system~\cite{fedscale-icml22} (\S\ref{s:impl}). 
FedScale supports plain  federated learning without differential privacy or protection against 
    a byzantine server. However, it is flexible, 
    allows a programmer to specify models in the PyTorch framework~\cite{pytorch}, and 
    includes a host of models and datasets, 
with model sizes ranging from 49K to 3.9M, for easy evaluation.
Our evaluation of 
    \sys's prototype (\S\ref{s:eval})
shows that \sys trains models with comparable accuracy to FedScale, in particular, the plain FedAvg algorithm 
in FedScale (\S\ref{s:eval:FedScale}).
\Sys also improves overhead relative to Orchard by up to five orders of magnitude, to a point where 
the overhead is low to moderate.
For instance, for
    a 1.2M parameter CNN  on the FEMNIST dataset~\cite{reddi2020adaptive, cohen2017emnist}, and 
    for a total population of $10^{9}$ devices where $10^{4}$ contribute updates per round, 
    a \sys device requires 15~ms of \cpu time and 3.12~KiB of network transfers per
round. Occasionally (with a probability of $10^{-5}$ in a round) this overhead increases 
when a device contributes updates, to a moderate 13.4 minutes of latency on a six-core processor 
and 234~MiB in network transfers. 

Prior to \sys, one must choose two of the three properties of 
high accuracy, rigorous privacy guarantee, and low device overhead.
With \sys, one can train models in a federated manner 
with a balance across these properties, 
at least for a particular class of federated learning 
algorithms. Thus, \sys's main contribution is
that it finally shows a way for a data analyst %
to train models while asking the data providers 
to place no trust in the analyst or their company.

\section{Problem and background}
\label{s:problem}

This section outlines the problem and gives a short background
on Orchard~\cite{roth2020orchard} that forms both a baseline and 
an inspiration for \sys.

\subsection{Scenario and threat model}
\label{s:problem:scenario}
\label{s:problem:threatmodel}
We consider a scenario consisting of a data analyst and a large 
number of mobile devices, e.g., hundreds of million.
The analyst, perhaps at a large company such as Google, 
is interested in learning a machine learning model 
over the data on the devices. For instance, the analyst may want to 
train a recurrent neural network (RNN) to provide auto-completion suggestions for the android 
keyboard~\cite{hard2018federated}.%

One restriction we place on this scenario is  
\changebars{that the training must be done}{that the analyst must perform the training} in a federated manner.
(We refer the reader to prior work~\cite{tan2021cryptgpu,knott2021crypTen} %
for a discussion on training in the centralized, 
non-federated model.)
As noted earlier (\S\ref{s:intro}), federated learning proceeds in rounds, 
where in each round
devices download the latest model parameters from a server,  generate
updates to these parameters locally, and send the updates to the server. 
The server aggregates \changebars{}{all} the model updates. 
This repeats until \changebars{the model}{either the model converges
or} achieves a target accuracy.

In this scenario, a malicious server, or even a malicious device, can execute many attacks.
For instance, a malicious server can infer the training data of a device from the updates contributed by the device~\cite{zhu2019deep,melis2019exploiting,shokri2017membership}.
Similarly, a malicious device that receives model parameters from the server can execute an inference attack to learn another device's input.%

We assume the strong \emph{OB+MC} threat model from Orchard. 
The server is honest-but-curious most of the time but occasionally byzantine (OB), while 
the devices are mostly correct (MC), but a small fraction can be
malicious.
The rationale behind the occasionality of the server's maliciousness is that the server's 
operator, e.g., Google, is reputed and 
subject
to significant scrutiny from the press and the users, and thus unlikely to be byzantine for long.
However, it may occasionally come under attack, e.g., from a rogue employee.
The rationale behind the smallness of the fraction of malicious devices 
is that with billions of devices, 
it is unlikely that an adversary will 
control more than a small percentage. 
For instance, for a billion devices, only controlling 3\%  
is already significantly larger than a large botnet. We further assume that 
a configurable percentage of honest devices may be offline during any given round of training.%
    
\subsection{Goals}
\label{s:problem:goals}
Under the OB+MC threat model, we want our system to meet the following goals.

\emparagraph{Privacy.}
It must guarantee the gold standard 
definition of privacy, i.e., differential privacy (DP)~\cite{dwork2011firm,dwork2006calibrating,dwork2014algorithmic,abadi2016deep}.
Informally, a system offers 
    DP for model training 
    if the probability of learning a particular set of model parameters %
    is (approximately) independent of whether a  device's input is 
    included in the training. 
This means that DP prevents inference attacks~\cite{Naseri2022LocalAC} where a particular device's input is revealed, 
    as models are (approximately) independent of its input.

\emparagraph{Accuracy.}
During periods when the server or the devices that contribute in a round are not byzantine, our system must produce models with accuracy comparable to models trained via plain federated learning. That is, we
want the impact of differential privacy to be low. Further, we want the system to mitigate a malicious device's impact on accuracy and prevent it from supplying arbitrary updates.

\emparagraph{Efficiency and scalability.} 
We want the system to support models with a large number of
parameters while imposing a low to moderate device-side overhead.
For the former, a reference point
is the android 
keyboard auto-completion model (an RNN) with 1.4M parameters~\cite{hard2018federated}.
For the device overhead, if a device participates
regularly in training, e.g., in every round, then it 
    should incur no more than a few seconds of 
    \cpu and a few MiBs in network transfers per round. 
However, we assume that devices can tolerate occasional amounts of additional work, 
contributing tens of minutes of \cpu and a few hundred MiBs in network transfers.

\subsection{Possible solution approaches}
\label{s:problem:solutions}
Meeting all of the goals described above is quite challenging. 
For illustration, consider the following solution approaches. 

\emparagraph{Local differential privacy.}
One option to guarantee differential privacy 
is to pick a federated learning algorithm that incorporates\changebars{}{ a particular type of differential privacy
called} local differential privacy (LDP)~\cite{duchi2013local}.
In LDP-based federated learning, 
devices add statistical noise to their updates before uploading them to the server.
The added noise protects updates against a malicious server which now cannot execute
an inference attack, but LDP
significantly degrades model accuracy relative to plain 
federated learning
as every device must add noise. For instance, 
the LDP-FL~\cite{ijcai2021-217} system trains a VGG model over \changebars{CIFAR10}{the CIFAR10 dataset} with 
10\% accuracy\changebars{ compared to 62\% with plain federated learning}{, while the same model achieves 62\% accuracy with plain federated learning}.

\emparagraph{Trusted server.}
\changebars{One alternative}{The alternative to LDP} is to use central differential privacy~\cite{abadi2016deep}, where
a central entity adds a smaller amount of noise to device updates to
ensure differential privacy. 
\changebars{This approach mitigates the accuracy issue, but provides weak to no privacy as the central entity sees devices' updates.}{Though this approach mitigates the accuracy issue, it provides weak to no privacy as
the central entity  %
sees devices' updates and can perform inference attacks.}

\emparagraph{Server-side secure multiparty computation (MPC).}
One way to reduce trust in the central server is to break it down into multiple non-colluding pieces, 
e.g., three servers, that run in separate
administrative domains. Then, one would run a secure multi-party computation protocol~\cite{yao1982protocols,goldreich2019play}
among these servers such that they holistically perform the required computation (noise generation, addition, etc.) 
while no individual 
server sees the input or intermediate state of
the computation. The problem is that we must still put significant trust in the server---that an adversary cannot compromise, say, two administrative domains.

\emparagraph{Large-scale MPC.}
One can remove trust in the server by instead running
MPC among the devices (essentially the devices perform the server's work). The problem now shifts to 
efficiency and scalability: general-purpose MPC protocols
are expensive and do not scale well with the number of participants~\cite{scaleMamba,damgaard2012multiparty}. %
Indeed, scaling MPC to a few hundred or thousand participants 
is an active research area~\cite{gordon2021more, ben2021large}, 
let alone hundreds of millions of participants.

\emparagraph{State-of-the-art: Orchard.}
\begin{figure}[t]
 {\includegraphics[width=\columnwidth]{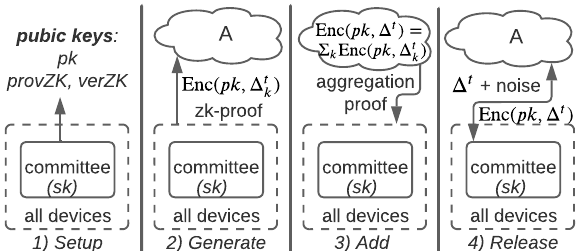}}
\caption{An overview of Orchard~\cite{roth2020orchard}. $\Delta_k$ denotes $k$-th device's update. The superscript $t$ denotes 
the  round number. Orchard runs the four phases of setup, generate, add, and release for every round.}
\label{f:orchard}
\label{fig:orchard}
\end{figure}

Orchard~\cite{roth2020orchard} takes a middle ground. It runs small(er)-scale MPC among devices 
while assuming an (occasionally) byzantine server.
In particular, it forms a \emph{committee} of a few tens of devices picked randomly
from the entire population of devices; this committee then runs MPC among its members.
Figure~\ref{f:orchard} shows an overview of Orchard. Orchard
supports the full-batch gradient descent algorithm where all devices contribute 
updates in a round.
For every round, 
Orchard runs the following four phases.

In the setup phase, Orchard samples the committee. 
Its members use MPC
to generate keys for two cryptographic primitives: additive homomorphic encryption (AHE)~\cite{rivest1978data,fan2012somewhat} and zero-knowledge proof (ZK-proof)~\cite{groth16on,goldwasser2019knowledge}.

In the generate phase, devices 
download the latest model parameters from the server. 
After local training, they encrypt their model updates (denoted by $\Delta$ in Figure~\ref{f:orchard}) 
and generate a ZK-proof 
to prove to the server that the ciphertexts are well-formed and the model updates are bounded. The encryption hides the updates from the byzantine server and the proof limits the impact of malicious devices (they cannot supply arbitrary updates and thus ruin the model accuracy).

In the add phase, the server homomorphically adds the encrypted updates. 
The server also generates proof that it performed the addition correctly so that all devices 
can collaboratively verify the addition. This verification is necessary to prevent a byzantine server
from launching subtle attacks to break DP. (We will discuss these attacks further in \S\ref{s:design:add}.)

In the release phase, the committee from the setup phase 
uses MPC to decrypt the output ciphertexts from the add phase. 
The committee also generates and adds the DP noise to the output, before releasing it,
to guarantee (central) DP. 

The challenge with Orchard is that even though it uses MPC at a smaller scale, the MPC overhead is still high. 
First, in the setup phase, committee devices generate fresh keys for each round, and generating one AHE key pair inside general-purpose MPC requires $\approx$180 seconds of \cpu and 1~GiB of network transfers. Second, the overhead of the add phase to verify the server's work grows linearly with the model size and becomes unpragmatic as soon as the model has a few hundred thousand parameters.
Third, in the release phase, committee devices decrypt ciphertexts and generate DP noise inside general-purpose MPC, costing, for example, $\approx$2600 seconds of \cpu and $\approx$\changebars{38}{19}~GiB of network transfers per device for a model with just 4K parameters.

\smallskip

In general, providing high accuracy, differential privacy, and device 
efficiency simultaneously in a threat model where there is no trusted party 
proves incredibly challenging.

\section{Overview of \sys}

The high-level idea in \sys is to focus on a specific type of federated learning algorithms comprising DP-FedAvg~\cite{mcmahan2017learning}, DP-FedSGD~\cite{mcmahan2017learning}, and DP-FTRL~\cite{kairouz2021practical}. These algorithms have similar characteristics; for instance, they all sample noise for differential privacy from a Gaussian distribution. To keep \sys easier to explain and understand, we take the most popular DP-FedAvg as the canonical algorithm and describe \sys in its context.

\subsection{DP-FedAvg\changebars{ without amplification}{}}
\label{s:dpfedavg}
\label{s:overview:dpfedavg}
\begin{figure}[t]
\begin{small}
{

\newcommand{\Stateindented}{\State\hspace*{5mm}}

\begin{algorithmic}[1]

    \Function{\textbf{\main}}{:}
    \State \textit{parameters}
            \Stateindented device selection probability $q\in (0,1]$
            \Stateindented DP noise scale $z$
            \Stateindented total \# of devices $W$
            \Stateindented clipping bound on device updates $S$
            
            \smallskip
            \State Initialize model $\theta^0$, DP privacy budget accountant $\mathcal{M}$
            
            \For{each round $t = 0, 1, 2, \ldots$} \label{l:dpfedavg:forloop}
                \State $C^{t} \gets$ (sample users with probability $q$) \label{l:dpfedavg:sampleclients}
                \For{each user $k \in C^{t}$}
                
                    \State $\Delta_{k}^{t} \gets \Call{\userupdate}{k, \theta^{t}, S}$ \label{l:dpfedavg:userupdate}
                \EndFor
                \LineComment{Add updates and Gaussian DP noise with $\sigma = zS$} %
                \State $\Delta^{t} \gets \sum_k \Delta_{k}^{t}+ \mathcal{N}(0,I\sigma^2)$ \label{l:dpfedavg:aggregateupdates} %
                \State $\theta^{t+1} \gets \theta^{t} + (\Delta^t/(qW))$ \Comment{Update model} \label{l:dpfedavg:updateglobalmodel}
                    
                \State Update $\mathcal{M}$ based on noise scale $z$ and parameter $q$ \label{l:dpfedavg:updateprivacybudget} %
            \EndFor
    \EndFunction
    
    \medskip
    \Function{\textbf{\userupdate}}($k$, $\theta^0$, $S$)
     \State \textit{parameters} $B, E, \eta$ \Comment{$\eta$ is learning rate}
     \smallskip
     \State $\theta \gets \theta^0$
     \For{each local epoch $i$ in $1$ to $E$}
        \State $\mathcal{B} \gets$ ($k$'s data split into size $B$ batches)
        \For{batch $b \in \mathcal{B}$}
            \State $\theta \gets \theta - \eta \nabla \ell (\theta; b)$ \Comment{$\ell$ is loss fn. (model err.)} \label{l:dpfedavg:gradientdescentstep}
            \State $\theta \gets \theta^0 + Clip(\theta - \theta^0, S)$ \label{l:dpfedavg:clipping}
        \EndFor
     \EndFor
     \State \Return $\Delta_k = \theta - \theta^0$ \Comment{Already clipped}
    
    \EndFunction
    
\end{algorithmic}
}
\end{small}
\caption{%
Pseudocode for the DP-FedAvg algorithm. $Clip(\cdot, S)$ scales its input vector 
such that its norm (Euclidean distance from the origin) is less than $S$. $\mathcal{M}$ is the privacy budget accountant of Abadi et al.~\cite{abadi2016deep} that tracks the values of the DP parameters $\epsilon$ and $\delta$.}
\label{f:dpfedavg}
\label{fig:dpfedavg}
\end{figure}

\begin{figure*}[t]
  \centering{\includegraphics[width=\textwidth]{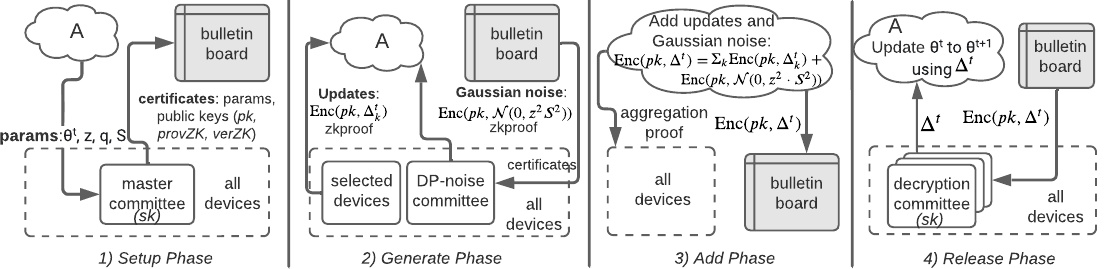}}
\caption{%
An overview of \sys's architecture and the four phases of its protocol.}
\label{f:overview}
\label{fig:overview}
\end{figure*}

DP-FedAvg proceeds in discrete rounds (Figure~\ref{f:dpfedavg}). 
In each round $t$, it samples a small subset of user devices using a probability parameter $q$ (line~\ref{l:dpfedavg:sampleclients}), and asks the sampled devices to provide updates to the global model parameters (line~\ref{l:dpfedavg:userupdate}). 
The devices locally generate the updates before clipping them by a value $S$ and uploading them (line~\ref{l:dpfedavg:clipping}); this clipping is necessary for differential privacy and it bounds the norm (sensitivity) of a device's update. %
DP-FedAvg then aggregates these updates (line~\ref{l:dpfedavg:aggregateupdates}) and (separately) adds noise sampled from a Gaussian distribution. 
The standard deviation of the Gaussian distribution depends on 
a noise scale parameter $z$ and the clipping bound $S$; both are 
input parameters for DP-FedAvg. %
Finally, DP-FedAvg updates a privacy accountant $\mathcal{M}$ that computes, based on 
the noise scale $z$ and sampling probability $q$, 
two parameters $\epsilon$ and $\delta$ associated with differential privacy (line~\ref{l:dpfedavg:updateprivacybudget}).
These parameters capture the strength of the guarantee: how much 
the model parameters learned after a round vary depending on a device's input.
A lower value of $\delta$ and $\epsilon$ is desirable, and the literature
recommends ensuring  
that $\epsilon$ stays close to or below $1$, and $\delta$ is less than $1/W$, 
where $W$ is the total number of devices~\cite{mcmahan2017learning}.

DP-FedAvg has three characteristics that are crucial for \sys.
The first  is the sampling of devices (lines~\ref{l:dpfedavg:sampleclients} to~\ref{l:dpfedavg:userupdate} in Figure~\ref{f:dpfedavg}).
For instance, the sampling parameter $q$ could be $10^{-5}$ such that 
10{,}000 out of, say, 
$10^{9}$ total devices contribute updates in a round.
The second characteristic is that the noise
is sampled from a Gaussian distribution whose standard deviation $\sigma$ is predetermined 
(set before the algorithm is run). 
This is in contrast to other DP algorithms that utilize techniques such as the Sparse Vector Technique (SVT) that
generate noise depending on the value of the updates~\cite{roth2010interactive,dwork2009complexity}.
The third characteristic is averaging of updates: 
    DP-FedAvg simply adds updates and noise (line~\ref{l:dpfedavg:aggregateupdates} in Figure~\ref{f:dpfedavg}) 
    rather than combining them
    using a more complex function. \Sys heavily leverages these characteristics.

Finally, we remark that \sys can support DP-FedAvg only without the amplification assumption for DP.
This is because the adversary (the byzantine server) can observe all traffic and knows which devices contribute updates for training. 
In contrast, the amplification assumption requires the server to be oblivious to the contributors, which in turn improves the privacy budget. We leave 
the addition of expensive oblivious approaches (which hides who is contributing updates besides hiding the updates themselves) to future work.

\subsection{Architecture of \sys}
\label{s:arch}
\label{s:overview:arch}
\Sys borrows two system components from Orchard: 
an \emph{aggregator}
and a public
\emph{bulletin board} (Figure~\ref{f:overview}).
The aggregator runs server-side inside a data center and therefore 
    consists of one or more
powerful machines.
    Its main role is to combine updates from user devices
        without learning their content.
        The bulletin board is an immutable append-only log.
The aggregator (which is potentially malicious) and the devices use the bulletin board 
to reliably broadcast messages 
and store states across rounds, e.g., the latest values of 
DP parameters $\epsilon$ and $\delta$.
Like Orchard~\cite{roth2020orchard}, \sys
        assumes that 
    free web services such as Wikipedia, or a public block-chain
    could serve as the bulletin board.

Like Orchard, \Sys also consists of \emph{committees} of devices, but 
instead of a single committee as in Orchard, 
\Sys has three types of committees tailored to the needs of DP-FedAvg (and similar algorithms).
A \emph{master committee} handles system setup, 
    including key generation for cryptographic primitives.
A \emph{DP-noise committee} handles
Gaussian noise generation.
And multiple \emph{decryption committees}
perform decryption operations
    to release updates to the global model
parameters at the end of a training round.
\Sys samples each committee afresh
    each round, dividing the committee
workload across the large population of devices.

An architecture with separate committee types is deliberate and
crucial.
It helps tailor a committee's protocol to its tasks
to significantly improve efficiency. 
Besides, the use of
multiple committees of the same type, i.e., 
multiple decryption committees, 
helps \sys scale with model size
as each committee works on
a subset of model  parameters.

Notably, the ability to use separate committee types is 
possible
only because of the specifics of DP-FedAvg. 
For instance, the fact that 
Gaussian noise generation does not depend on the value of the updates
allows \sys to separate the DP-noise committee from the decryption committees.

\subsection{Protocol overview of \sys}
\label{s:protocoloverview}
\label{s:overview:protocol}

To begin training a model, a data analyst supplies
input parameters
(the model architecture, the initial model parameters, and the other input 
parameters for DP-FedAvg)
to the aggregator. The aggregator then initiates 
a \emph{round-based protocol} consisting of discrete rounds.
In each round, it executes one iteration of the 
    for loop
    in the $\main$ procedure of DP-FedAvg (line~\ref{l:dpfedavg:forloop} 
    in Figure~\ref{f:dpfedavg}).
Each round further consists of the four phases of
    \emph{setup}, 
    \emph{generate},
    \emph{add},
    and \emph{release} (Figure~\ref{f:overview}).

In the setup phase, the aggregator samples the various committees for the round. 
The master committee then
receives and validates the input parameters, and generates keys for an AHE and a
ZK-proof scheme. \Sys's setup phase is similar to Orchard (\S\ref{s:problem:solutions}) with
the key difference that \sys's master committee uses techniques  to reuse keys across rounds rather than generating
them fresh for each round using MPC. %

Next, in the generate phase, (i) 
    devices select themselves 
to generate updates for the round,  and (ii)
the DP-noise committee generates the Gaussian noise for DP.
Both types of devices
use techniques to perform their work efficiently.
For instance, the DP-noise committee generates noise in a distributed manner while avoiding MPC.

Next, in the add phase, 
the aggregator adds the model updates to the Gaussian noise without 
learning the plaintext content of either of them. 
This is done through the use of the AHE scheme. 
The entire population of devices collectively verifies the aggregator's work. 
Again, the key is efficiency for the devices, for which the aggregator and the devices
use a new verifiable aggregation protocol.

Finally, in the release phase, each decryption committee receives the 
secret key for the AHE scheme from the master committee and 
decrypts a few ciphertexts from the add phase.
The key point is that 
a decryption committee avoids general-purpose MPC 
by using a specialized decryption protocol.

\section{Design of \sys}
\label{s:design}

We now go over the design details of \sys
phase-by-phase.
The main challenge in each phase is keeping the device
overhead low while protecting against the malicious aggregator
and the malicious subset of devices. 
We highlight these challenges, 
and \sys's key design choices and techniques.

But before proceeding, we briefly discuss committee formation, which is common to 
multiple phases.
To form committees, \sys uses
the sortition protocol from Orchard (which in turn used Algorand's
protocol~\cite{gilad2017algorand}).
This protocol relies on a publicly verifiable
source of randomness so that the results of the election
    are verifiable by all devices.
At the end of the protocol, 
    the aggregator publishes
    the list (public keys) of the committee members 
    by putting it
    on the bulletin board. An important aspect 
of committee formation is  committee size and the number of malicious devices in a committee:%
provision of a larger number of
malicious devices $A$ relative to the committee size $C$ increases costs but ensures 
higher resiliency. 
Like Orchard, \sys makes a probabilistic argument~\cite{roth2019honeycrisp}
to select $C$ and $A$ such that the 
probability of the number of malicious devices exceeding $A$ is small.
For example, if 
  the overall population%
        contains up to $f = 3\%$ malicious devices (\S\ref{s:problem:threatmodel}), 
    then the probability that a randomly sampled subset of $C = 45$
    devices contain more than $A = 2C/5 = 18$ malicious 
    devices is less than $9.6 \cdot 10^{-14}$.

\subsection{Setup phase}
\label{s:design:setup}

Much of \sys's setup phase is similar to Orchard. 
During this phase, (i) the aggregator
    samples
    the master committee, which then
    (ii) receives and validates inputs for the round
     (i.e., receives model parameters $\theta^t$ for the current round $t$,
the device selection probability $q$, 
noise scale $z$, and clipping bound $S$, and generates
    new values of the DP parameters $\epsilon, \delta$),
    and (iii) generates keys for cryptographic 
    primitives (\S\ref{s:protocoloverview}).
We do not focus on the first two pieces to avoid repetition with Orchard, but 
    include them
    in the supplementary material for completeness (Appendix~\S\ref{s:appendix:design}).
Instead, the
    key challenge in \sys is the overhead of key generation.%

Recall (\S\ref{s:problem:solutions}) that Orchard uses MPC among the master committee members to correctly run the key generation function
and ensure that 
even if the malicious members of the committee collude,
they cannot recover the AHE secret key. 
The overhead of this MPC is high: $\approx$1~GiB of 
network transfers and 180 seconds of \cpu time per committee device.%
How can this overhead be reduced?

One idea~\cite{roth2021mycelium} is to reuse keys across rounds rather than generate them afresh 
for each round. Indeed, this is what \sys does: the master committee in 
round 1 generates the keys and 
shares them with the committee for the next round, and this committee then shares the keys with the committee for the third round, and so on. 
But one has to be careful.

Consider the following attack. 
Say that the malicious aggregator receives a victim device $k's$ update $Enc(pk, \Delta_{k}^{t})$ in round $t$.
Then, in the next round $t+1$, 
the aggregator colludes with a malicious device in the overall population to use $Enc(pk, \Delta_{k}^{t})$ as the device's update. 
This attack enables the aggregator to 
violate differential privacy as
the victim device's input does not satisfy the required clipping bound $S$ in round $t+1$ 
due to its multiple copies 
(\S\ref{s:dpfedavg}). Orchard does not suffer from this attack as it 
generates fresh keys in each round: the ciphertext for round $t$ 
decrypts to a random message with round $t+1$'s key. However, prior 
work that reuses keys in this manner 
(in particular, Mycelium~\cite{roth2021mycelium}) does suffer from this attack.

Thus, \sys must apply the reuse-of-keys idea with care. 
\Sys adjusts the generate and add phases of its protocol (\S\ref{s:overview:protocol})
to prevent 
the aforementioned attack. We are not in a position yet to describe these changes, but we will
detail them shortly when we describe these other phases (\S\ref{s:design:generate}, \S\ref{s:design:add}).
Meanwhile, 
    the changes in the setup phase relative to Orchard are the following: 
    for the AHE secret key $sk$,
    \sys implements an efficient verifiable secret redistribution
    scheme~\cite{gupta2006extended,roth2021mycelium}
    such that committee members at round $t+1$ 
    securely obtain the relevant shares of the key 
from the committee at round $t$.
For the public keys (AHE public key $pk$, and both the ZK-proof public proving and verification keys),
the committee for round $t$ signs a certificate containing these keys and uploads it 
to the bulletin board, and the committee for round $t+1$ downloads it from the board.

The savings by switching from key generation to key resharing are substantial for the network, with 
a slight increase in \cpu.
While the MPC solution incurs $\approx$1~GiB of network transfers and 180 seconds
of \cpu time per committee device, 
key resharing requires 125~MiB and 187 seconds, respectively ({\S\ref{s:eval:Orchard}}).
The \cpu is higher because key resharing requires
certain expensive field exponentiation operations~\cite{gupta2006extended}.

\subsection{Generate phase}
\label{s:design:generate}

Recall from \S\ref{s:overview:protocol} that during this phase 
(i) \sys must pick a subset of devices
to generate updates to the model parameters,
(ii) the DP-noise committee must generate
Gaussian noise for differential privacy, 
and (iii) both types of devices must 
encrypt their generated data (updates and noise).

\emparagraph{Device sampling for updates.}
One design choice is to ask the aggregator to sample devices that will contribute updates.
The problem with this option is that the (malicious) aggregator may choose the devices non-uniformly; for instance, it may pick
an honest device more often than the device should be picked, violating differential privacy.
An alternative is to ask the devices
to sample themselves with probability $q$ (as required by DP-FedAvg; line~\ref{l:dpfedavg:sampleclients} in Figure~\ref{f:dpfedavg}). 
But then a malicious device may pick
itself in every round, which would allow it to significantly affect model
accuracy.

\sys adopts a hybrid and efficient design in which devices sample themselves but
the aggregator verifies the sampling.
Let $B^{t}$ be a publicly verifiable source of
randomness for round $t$; this is the same randomness that is used in the sortition protocol
to sample committees for the round. Then, each device $k$
with public key $\pi_k$ computes 
$PRG(\pi_k || B^{t})$, where $PRG$ is a pseudorandom generator.
Next, the device
scales the PRG output to a value between 0 and 1, and checks if the result is less than $q$. 
For instance, if the PRG output is 8 bytes,
then the device divides this number by $2^{64} - 1$.
If selected, the device runs the $\userupdate$ procedure
(line~\ref{l:dpfedavg:userupdate} in Figure~\ref{f:dpfedavg}) to generate updates
for the round.
This approach of sampling is efficient as
devices only perform local computations.

\emparagraph{Gaussian noise generation.}
The default option is to make the 
DP-noise committee generate the noise using MPC,
but as noted
several times in this paper, this option is expensive. 
Instead, \sys adapts prior work~\cite{truex2019hybrid} on distributed Gaussian noise generation. 
The Gaussian distribution has the property that if an element 
sampled from $\mathcal{N}(0, a)$ is added to another element sampled from 
$\mathcal{N}(0, b)$, then the sum is a sample of $\mathcal{N}(a+b)$~\cite{truex2019hybrid,xu2019hybridalpha,dwork2006our}.
This works well for
the simple case when all $C$ committee members of the DP-noise committee 
are honest. 
Given the standard deviation of the Gaussian distribution, $\sigma = z \cdot S$, the devices can 
    independently compute their additive share.
That is, 
to generate samples from $\mathcal{N}(0, I \sigma^2)$ (line~\ref{l:dpfedavg:aggregateupdates} in Figure~\ref{f:dpfedavg}), each
committee member can sample its share of the noise from the distribution
$\mathcal{N}(0,I\frac{\sigma^2}{C})$.

The challenge in \sys is therefore:
how do we account for the $A$ malicious devices in the DP-committee?
These devices may behave arbitrarily and may thus generate 
either no noise or large amounts of it. 
Adding unnecessary noise hurts accuracy, not privacy. 
In contrast, failing to add noise may violate privacy. 
We thus consider the worst case in which malicious users fail to 
add any noise and ask honest devices to compensate. 
Each honest client thus samples its noise share from the 
distribution $\mathcal{N}(0,I\frac{\sigma^2}{C - A})$.\footnote{\Sys can further compensate for 
honest-but-offline devices. Say, for e.g., that $B$
of $C-A$ honest devices must be provisioned to be offline.
Then, \sys subtracts $B$ from $C-A$ to get the number of honest-but-online devices.}

This algorithm generates noise cheaply without expensive MPC. 
The downside is that it may generate 
more noise than necessary, hurting accuracy. To mitigate this risk, we carefully choose
the committee size to minimize the ratio of additional noise. Specifically, we choose $C$ 
to keep the ratio $(C - A) / C$ close to 1. 
For instance,
instead of picking a committee containing a few tens of devices similar to the master committee, 
we pick a somewhat larger DP-noise committee: $(A, C) = (40,
280)$.\footnote{Using a probabilistic argument for committee size selection as
before (\S\ref{s:design}), 
if $f=3\%$ devices in the overall population are malicious, then 
the chances of sampling 280 devices with more than 1/7th malicious is 
$4.1 \cdot 10^{-14}$.}

\emparagraph{Encryption and ZK-proofs.}
Once the devices generate their updates or shares of the Gaussian noise, they encrypt
\changebars{the}{their generated} content using the public key of the AHE scheme 
to prevent the aggregator from learning the content. Further, they
 certify using a ZK-proof scheme 
 that the encryption is done correctly and the data being 
encrypted is bounded by the clipping value $S$ (so that malicious devices may not supply arbitrary updates).
This encryption and ZK-proof generation is same as in Orchard, but \sys requires
additional changes. 
 Recall from the setup phase that 
 \sys must ensure a ciphertext generated in a round is used 
    only in that round, to prevent complications due to 
    reuse of keys (\S\ref{s:design:setup}).
To do this, each device
    concatenates the round number $t$ (as a timestamp)
to the plaintext message before encrypting it.
Further, the ZK-proof includes additional constraints
that prove that a prefix of the plaintext 
message equals the current 
round number.

\begin{figure*}[t]
\hrule
\medskip
{
\textbf{Add phase protocol of \sys}

\vspace{-0.02ex}
\textit{Commit step}
\vspace{-0.5ex}
\begin{myenumerate2}

    \item %
    \label{l:addphase:generatecommitments}
    Each device (with public key $\pi_i$)
    that generates ciphertexts $(c_{i1}, \ldots, c_{i\ell})$
    and corresponding ZK-proofs $(z_{i1}, \ldots, z_{i\ell})$
    in the generate phase does the following: for each $j \in [1,\ell]$, 
    generates a commitment to $(\pi_i, c_{ij})$, namely $t_{ij} = Hash(r_{ij} ||c_{ij}||\pi_i)$, 
    where $r_{ij}$ is a random 128-bit nonce. The device sends
    $(\pi_i,t_{i1},...,t_{i \ell})$ to the aggregator $\mathcal{A}$. 
   
    \item $\mathcal{A}$ checks that $PRG(\pi_i || B^{t}) \leq q$
or $\pi_i$ is a member of the DP-noise committee. Otherwise,  
    it ignores the message.
    
    \item %
    \label{l:addphase:merkletreeforcommitments}
    For each $j\in [1,\ell]$, $\mathcal{A}$ sorts pairs 
    $(\pi_i,t_{ij})$ by $\pi_i$ to form an array of tuples 
    $Commit_j$. $\mathcal{A}$ then
    generates a Merkle tree $MC_j$ from array
    $Commit_j$ and publishes the root to the bulletin board.
    
    \item %
    Each device 
        from step~\ref{l:addphase:generatecommitments} above
        sends $(\pi_i,c_{ij},r_{ij},z_{ij})$ to the aggregator, for each $j \in [1, \ell]$.
   
    \item $\mathcal{A}$ checks the messages. 
    If any 
    proof $z_{ij}$ or any commitment
    $t_{ij}=Hash(r_{ij}||c_{ij}||\pi_i)$ is incorrect, it 
    ignores the message. 
    
\end{myenumerate2}

\vspace{-0.4ex}
\textit{Add step}
\vspace{-0.5ex}
\begin{myenumerate2}
\setcounter{enumi}{5}

    \item 
    \label{l:addphase:createsummationtrees}
    For each $j\in [1,\ell]$, $\mathcal{A}$ computes a summation tree
$ST_j$.  The leaves are $ST_j(0,i)=(\pi_i,c_{ij},r_{ij},z_{ij})$
    if $\mathcal{A}$ got a correct message and $(\pi_i,\bot)$ otherwise. 
    Each non-leaf vertex has two children and a parent ciphertext is the sum of
its children ciphertexts.
    
    \item 
    \label{l:addphase:publishmerklesummationtrees}
For each $j\in [1,\ell]$, $\mathcal{A}$ serializes all vertices of
$ST_j$ into an array and publishes a Merkle tree $MS_j$, as well as the root of
the summation tree $ST_j$. 
    To each device that sent a correct leaf, $\mathcal{A}$ also sends a proof
that this leaf is in $MS_j$.

\end{myenumerate2}

\vspace{-0.4ex}
\textit{Verify step}
\vspace{-0.5ex}

\raggedright\hspace{0.4em} Every device in the system does the following:
\begin{myenumerate2}
\setcounter{enumi}{7}
    \item Retrieve the \# of leaf nodes $M'$ per summation tree from
$\mathcal{A}$.
Verify $M'\leq M_{max}$, where $M_{max}$ is a conservative bound on the total
\# of devices that contribute input, i.e., $W_{max}$ times $q$,
where $W_{max}$ is an upper bound on the total \# of mobiles.

    \item 
    \label{l:addphase:picksummationtrees}
    With probability $q$, select elements from 
    $[1,\ell]$ to form an index array $idx$ 
    containing indices of the trees to inspect.
   
    \item 
    \label{l:addphase:pickleafvertices}
    Choose a random $v_{init}\in [0,M'-1]$. 
    For each $j\in idx$, and 
        for each $i \in [v_{init},v_{init}+s)\mod M'$, 

    \begin{myenumerate}

        \item 
               verify that $ST_j(0,i)=(\pi_i,\bot)$ or
$(\pi_i,c_{ij},r_{ij},z_{ij})$, 
    and $\pi_i\le \pi_{i+1}$ (except if $i=M'-1$)

        \item 
    \label{l:addphase:checkzkproof}
    if $ST_j(0,i)\neq (\pi_i,\bot)$,
        verify that commitment $t_{ij}$ 
                    appears in $MC_j$, 
$t_{ij} = H(r_{ij} ||c_{ij} || \pi_i)$, and $z_{ij}$ is valid %

        \item 
        check $ST_j(0,i)$ is in $MS_j$

    \end{myenumerate}

    \item 
    \label{l:addphase:verifynonleaf}
    For each $j\in idx$, choose $s$ distinct non-leaf
vertices of $ST_j$. 
To reduce redundancy, this includes the $s/2$ vertices 
    whose children the device has already obtained in line~\ref{l:addphase:pickleafvertices}. 
For each such non-leaf vertex, verify
that its ciphertext equals the sum of the children ciphertexts,
    and that the vertex and its children are in $MS_j$.
\end{myenumerate2}

}
\hrule
\caption{%
\Sys's verifiable aggregation. This description
does not include the PIT optimization (described in text)
that applies to line~\ref{l:addphase:verifynonleaf}.}
\label{f:addphase}
\end{figure*}

\subsection{Add phase}
\label{s:design:add}

Recall that during the add phase (i) the aggregator adds ciphertexts containing device updates to those 
containing shares of Gaussian noise, (ii) the devices collectively verify the aggregator's addition (\S\ref{s:overview:protocol}).

This work during the add phase has subtle requirements. So first,
we expand on these requirements while considering a toy example with
 two honest and a malicious device. The first honest device's
input is $\encplain(pk, \Delta)$, where $\Delta$ is its update, while
the second honest device's input is
$\encplain(pk, n)$, where $n$ is the Gaussian noise.
For this toy example, first (\textbf{R1}), the aggregator must not omit $\encplain(pk, n)$ 
from the aggregate 
as the added noise would then be insufficient to protect $\Delta$ and guarantee 
DP. 
Second (\textbf{R2}), the aggregator must not let the malicious device use 
$\encplain(pk, \Delta)$ as its input. Relatedly, the aggregator itself must not modify 
$\encplain(pk, \Delta)$ 
to $\encplain(pk, k \cdot \Delta)$, where $k$ is a scalar, using the additively homomorphic
properties of the encryption scheme. The reason is that
these changes can violate the clipping requirement 
that a device's input is bounded by $S$ (e.g., $2 \cdot \Delta$ may be larger than $S$). 
And, third (\textbf{R3}), the aggregator must ensure that the above (the malicious device or the aggregator
copying a device's input) does not happen across rounds, as recall
that \sys uses the same encryption key in multiple rounds 
    (\S\ref{s:design:setup}).

One option to satisfy these requirements is to use %
the verifiable aggregation protocol of Orchard~\cite{roth2019honeycrisp} that is based on 
summation trees. The main challenge is resource costs. 
Briefly, 
in this protocol, the aggregator arranges the ciphertexts to be aggregated 
as leaf nodes of a tree, and publishes
the nodes of the tree leading to the root node. 
For example, the leaf nodes will be 
$\encplain(pk, \Delta)$ and $\encplain(pk, n)$, and the root node
will be $\encplain(pk, \Delta) + \encplain(pk, n)$, for the toy example above.
Then, devices in the entire population inspect parts of this tree: download
a few children and their parents and check that the addition is done correctly, that the leaf
nodes haven't been modified by the aggregator, and the leaf nodes that should be included 
are indeed included. 
The problem is that Orchard requires
a device to download and check about
$3 \cdot s$ nodes of the tree~\cite{roth2019honeycrisp,roth2020orchard}, where $s$ is a configurable parameter
whose default value is six. But for 
realistic models, each node is made of 
many ciphertexts (e.g., the 1.2M parameter CNN model requires $\ell = 293$ ciphertexts), 
and 18 such nodes add to 738~MiB.

\Sys improves this protocol using two ideas. 
First, \sys observes that the entire population of devices
that must collectively check the tree is massive (e.g., $10^{9}$). 
Besides, 
although the tree has bulky nodes with many ciphertexts, the total number of nodes
is not high due to sampling (e.g., only 10,000 devices contribute updates
in a round). 
Thus, \sys moves away from one summation tree with ``bulky'' nodes, to 
$\ell$ summation trees with ``small'' nodes, 
where $\ell$ is the number of ciphertexts
comprising a device's update (e.g., $\ell = 293$ for the 1.2M parameter model). 
Then, each device %
probabilistically selects a handful of trees, and
a checks few nodes within each selected tree.

Second, \Sys optimizes how a device tests 
whether the sum of two ciphertexts equals a third ciphertext.
\Sys 
recognizes that ciphertexts can be expressed as polynomials and
the validity of their addition can be checked  efficiently using a technique called polynomial identity testing (PIT)~\cite{schwartz1980fast,zippel1979probabilistic}.
Roughly, PIT says that the sum of polynomials can be checked by evaluating them at a random point and
checking the sum of these evaluations.
Using PIT, \sys replaces the ciphertexts at the non-leaf nodes
    of the summation trees with their much smaller evaluations at a random point.

We now describe \sys's protocol in detail, first without the PIT optimization, and then with it.

\emparagraph{Incorporating finer-grained summation trees.}
\sys's protocol has three steps: commit, add, and verify (Figure~\ref{f:addphase}). 
In the \emph{commit} step, all
devices commit to their ciphertexts before submitting 
them to the aggregator (line~\ref{l:addphase:generatecommitments}--\ref{l:addphase:merkletreeforcommitments} in
Figure~\ref{f:addphase}). The aggregator 
publishes a Merkle tree of these commitments to the
bulletin board. Committing before submitting ensures that a malicious device cannot copy and submit
an honest device's input (requirement \textbf{R2} above). 
Similarly, this design ensures that the aggregator cannot change a device's input (again requirement \textbf{R2}).

In the \emph{add} step, the aggregator adds the ciphertexts via 
summation trees. Specifically, if device updates have $\ell$
ciphertexts, the aggregator creates $\ell$ summation trees, one 
per ciphertext (line~\ref{l:addphase:createsummationtrees} in Figure~\ref{f:addphase}). 
The leaf vertices of the $j$-th tree are the $j$-th ciphertexts in the devices'
inputs, while each parent is the sum of its children ciphertexts,
and the root is the $j$-th ciphertext in the aggregation result. 
The aggregator publishes the vertices of the summation trees
on the bulletin board 
(line~\ref{l:addphase:publishmerklesummationtrees} in Figure~\ref{f:addphase}), 
allowing an honest device to check that its input is 
not omitted (requirement \textbf{R1} above).

In the \emph{verify} step, each device in the system 
selects $q \cdot \ell$ summation trees,
where $q$ is the device sampling probability %
(line~\ref{l:addphase:picksummationtrees} in
Figure~\ref{f:addphase}), and checks $s$ leaf nodes and $2s$ non-leaf nodes in each tree. ($s=6$ in our implementation.)
Specifically, the 
device checks that the leaf node ciphertexts are
committed to in the commit step (requirement \textbf{R2}), and the ZK-proofs
    of the ciphertexts are valid, e.g., the first part of the plaintext message
    in the ciphertexts equals the current round number (requirement \textbf{R3}).
For the non-leafs, the device checks that they sum to their children.

\emparagraph{Incorporating PIT.}
Checking the non-leaf vertices 
is a main source of overhead for the protocol
above. 
The reason is that even though each non-leaf is a single ciphertext, this ciphertext
is large: 
for the quantum-secure AHE scheme \sys uses (\S\ref{s:impl}), 
a ciphertext is 
    131~KiB, made of 
    two polynomials of $2^{12}$ coefficients each, where
    each coefficient is 16 bytes.

As mentioned earlier, \sys reduces this overhead by using 
polynomial identity testing (PIT)~\cite{schwartz1980fast,zippel1979probabilistic}.
This test says that given a $d$-degree polynomial $g(x)$ whose coefficients are
in a field
$\mathbb{F}$, one can test whether 
 $g(x)$ is a zero polynomial by picking a number $r\in \mathbb{F}$ 
uniformly and testing whether $g(r)==0$.
This works because a $d$-degree polynomial has at most $d$ solutions to
$g(x)==0$ and $d$ is much less than $|\mathbb{F}|$.

Using PIT, \sys replaces the ciphertexts at the non-leafs
    with their evaluations at a random point $r$.
Then, during the ``Verify'' step, a device
checks 
(line~\ref{l:addphase:verifynonleaf} in Fig.~\ref{f:addphase})
whether these evaluations (rather than ciphertexts) add up.
Thus, instead of downloading three ciphertexts with $2 \cdot 2^{12}$ 
    field elements each, 
    a device downloads 2 elements of $\mathbb{F}$ per ciphertext.

A requirement for PIT 
    is generation of $r$, which must be sampled
    uniformly from the coefficient field.
For this task,
\sys extends
    the master committee to
    publish an $r$ to the bulletin board 
    in the add step, using
a known protocol
to securely and efficiently
        generate a random number~\cite{damgard2006unconditionally, damgaard2012multiparty}.
     
\subsection{Release phase}
\label{s:design:release}
During the release phase, \sys must decrypt
the $\ell$ ciphertexts from the add phase, i.e.,
 the $\ell$ root nodes of the $\ell$ summation trees.
 The default, but expensive, option is to use MPC among the members of the decryption committees. 

\Sys addresses this efficiency challenge 
using known ideas and applying them; i.e., \sys's contribution
in this phase is not new techniques, but  the observation that existing ideas can be 
applied. Nevertheless, applying these ideas requires some care and work.

First, recall that \sys has multiple decryption committees (\S\ref{s:overview:arch}). Naturally, to reduce per-device work,
each  committee decrypts a few of the $\ell$ ciphertexts. A design question for \sys is
how many committees should it use. On the one hand, more committees are desirable (best case is $\ell$). However, more committees also mean that each has to be larger to ensure that
none of them samples more than $A$ out of $C$ malicious devices, breaking the threshold assumptions of a committee.
Meanwhile, a larger committee means more  overhead. In practice (\S\ref{s:eval}),
we take a middle ground and configure \sys to use ten decryption 
committees.

Second, \sys reduces each committee's work relative to the MPC baseline, using
a fast distributed decryption protocol to decrypt the ciphertexts~\cite{chen2019efficient}.
The use of this protocol is  possible as a decryption committee's task is only 
of decryption given how we formed and assigned work to different types of committees (\S\ref{s:overview:arch}).
This fast  protocol requires the committee devices to mainly perform local computations with little interaction 
with each other. The caveat is that for this protocol to be applicable, the committee members must
know an upper bound on the number of additive homomorphic operations on the ciphertexts
they are decrypting.\footnote{This bound is needed to add a ``smudging noise'' 
to the committee's decryption output to ensure that the output does
not leak information on the inputs to the aggregation~\cite{asharov2012multiparty}.} 
Fortunately, in \sys's setting this bound is known: 
    it is the maximum number of devices whose data the aggregator adds
    in the add phase ($M_{max}$ in Figure~\ref{f:addphase}).
The benefit of distributed decryption (and moving work outside MPC) meanwhile
is substantial.

\subsection{Privacy proof}
\label{s:design:proof}
\Sys's protocol
provides the required differential privacy guarantee (\S\ref{s:problem:goals}).
The supplementary material (Appendix~\ref{s:appendix:privacy_proof})
contains a proof. But, briefly, 
the key reasons are that (i) in the generate phase, honest devices sample themselves
to make sure that they are not sampled more than expected, (ii)
the verifiable aggregation protects these devices' input, and (iii)
key resharing and fast decryption protocols keep secret keys hidden.
\section{Prototype implementation}
\label{s:impl}
\begin{figure}[t]
 {\includegraphics[width=0.4\textwidth]{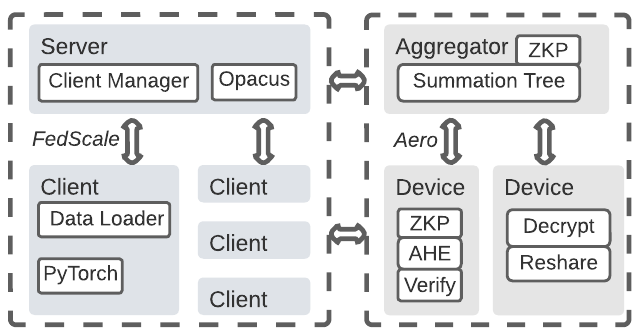}}
\caption{%
An overview of \sys's implementation.}
\label{f:impl}
\label{fig:impl}
\end{figure}

We implemented a prototype of \sys atop FedScale~\cite{fedscale-icml22}, which is a scalable 
system for federated learning capable of handling a large number of devices. By default, FedScale supports 
algorithms such as FedAvg and FedSGD 
(without differential privacy). Further, it allows a data analyst
to specify the model using the popular PyTorch framework.

\begin{figure}[t]
    \footnotesize
    \centering
    
    \begin{tabular}{
        @{}
        *{1}{>{\raggedright\arraybackslash}b{.06\textwidth}}%
        *{1}{>{\raggedleft\arraybackslash}b{.13\textwidth}}  %
        *{1}{>{\raggedleft\arraybackslash}b{.04\textwidth}}  %
        *{1}{>{\raggedleft\arraybackslash}b{.07\textwidth}}  %
        *{1}{>{\raggedleft\arraybackslash}b{.06\textwidth}}  %
        @{}
        }
        Dataset & Model & Size & FedScale & \sys\\
        \bottomrule

        FEMNIST~\cite{cohen2017emnist} & LeNet~\cite{lecun1995learning} & 49K & 75\% & 74\% \\
             & CNND~\cite{reddi2020adaptive} & 1.2M & 78\%	& 68\% \\
             & CNNF~\cite{mcmahan2017communication} & 1.7M & 79\% & 68\% \\
             & AlexNet~\cite{krizhevsky2012imagenet} & 3.9M & 78\% & 40\% \\
             
        \hline
        CIFAR10~\cite{Krizhevsky09learningmultiple} & LeNet~\cite{lecun1995learning} & 62K & 48\% & 48\% \\
             & ResNet20~\cite{he2016deep} & 272K & 59\% & 48\% \\
             & ResNet56~\cite{he2016deep} & 855K & 54\% & 35\%\\
        \hline
        Speech~\cite{warden2018speech} & MobileNetV2~\cite{howard2017mobilenets} & 2M & 57\% & 4\%\\
        \bottomrule

    \end{tabular}
    \caption{Test accuracy for different models after 480 rounds of training and differential privacy parameters $(\epsilon, \delta)$ set to (\changebars{5.03}{1.04}, $W^{-1.1}$). As shown later, increasing $\epsilon$ can recover the accuracy loss.}
    \label{fig:acc_models}
    \label{f:acc_models}
\end{figure}

\begin{figure*}[t]
\centering
\begin{subfigure}[t]{0.4\textwidth}{
\centering
\includegraphics[width=2.5in]{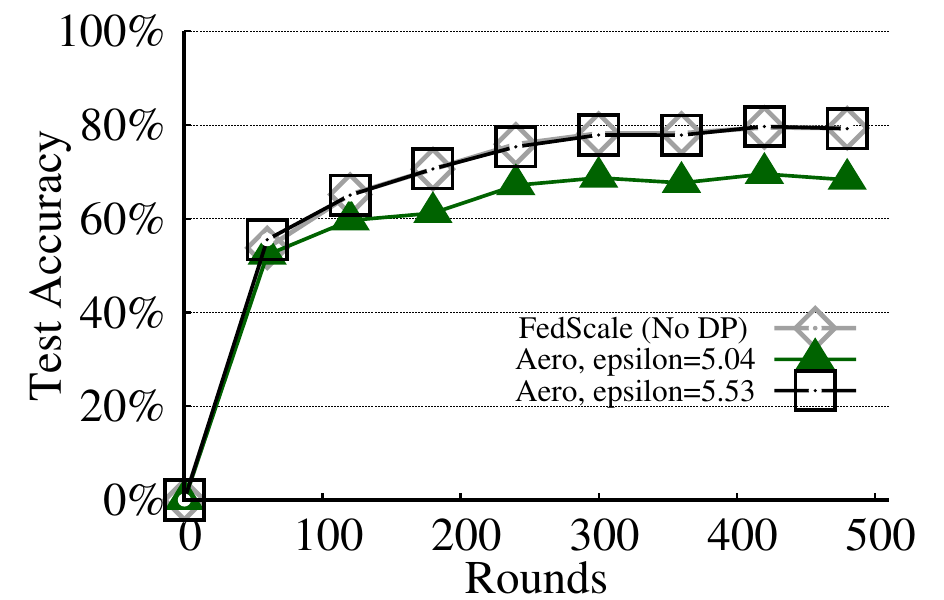}
\caption{CNNF}
\label{f:cnn_fedavg_acc}
}
\end{subfigure} %
\begin{subfigure}[t]{.4\textwidth}{
\centering
\includegraphics[width=2.5in]{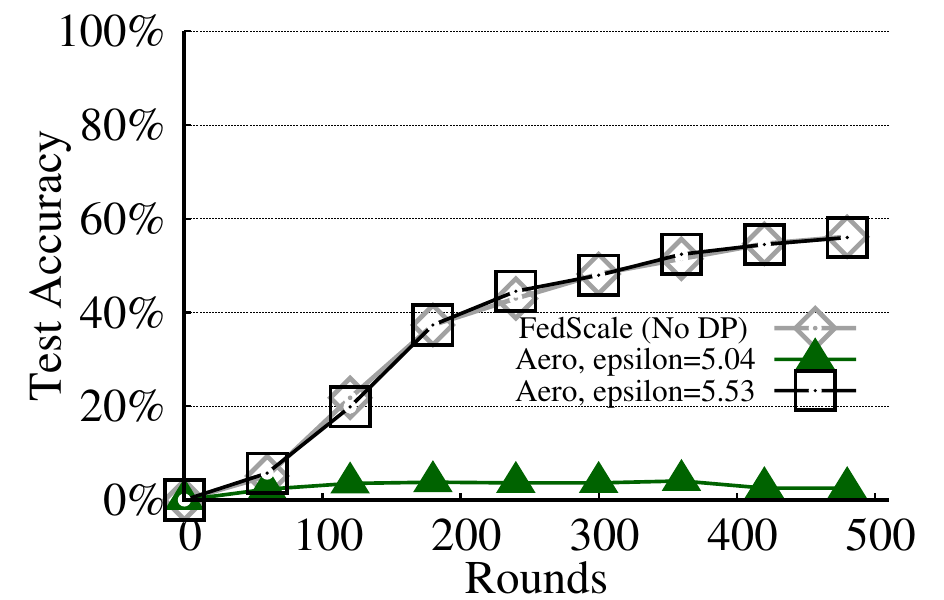}
\caption{MobileNetV2}
\label{f:mobilenet_acc}
} %
\end{subfigure}
\caption{%
Test accuracy versus rounds for \sys and FedScale for the 1.7M parameter CNNF model and the %
2M parameter MobileNetV2 model. %
} 
\label{f:acc_vs_rounds}
\label{fig:acc_vs_rounds}
\end{figure*}

Our \sys prototype extends FedScale in the following way (Figure~\ref{f:impl}).
First, it extends the programming layer of FedScale with Opacus~\cite{opacus}, which 
is a library that adjusts a PyTorch model to make it suitable for differentially
private federated learning; for instance, Opacus replaces the batch normalization
layer of a neural network with group normalization.
Second, our prototype
extends the device-side code of FedScale with additional components
needed for the various committees and phases in \sys (key resharing, Gaussian noise generation, 
verifiable aggregation, and distributed decryption; \S\ref{s:design}).
FedScale is written in Python while the code we added is in Rust; thus, 
we use PyO3 to wrap the Rust code with Python interfaces.
Third, our prototype extends the FedScale server-side code
with \sys's aggregator code and the code to coordinate
the various phases of \sys. 
In total, we added $\approx$4,300 lines of Rust
to FedScale.

Our prototype configures the cryptographic primitives for 128-bit security.
For additively homomorphic encryption, we use the BFV encryption scheme. We
set the polynomial degree in BFV to $2^{12}$ and use the default parameters
from Microsoft SEAL~\cite{seal}. 
For ZK-proofs, we use  ark\_groth16~\cite{arkgroth16}, which implements the
zkSNARK of Jens Groth~\cite{groth16on}.
\section{Evaluation}
\label{s:eval}

We evaluate \sys in two parts.
First, we compare it with plain federated learning, specifically, the FedScale system. This comparison
sheds light on the cost of privacy both in terms of model accuracy and resource consumption on the devices. 
Second, we compare \sys to Orchard, which is the state-of-the-art system 
for training models in a federated manner in the same threat model as \sys. 
This comparison helps understand the effectiveness of \sys's techniques in 
reducing overhead.
Our main  results are the following:

\begin{myitemize}

    \item \sys can train models with comparable  accuracy to
    FedScale (plain federated learning). For instance, for a CNN model over the FEMNIST dataset, 
    \sys produces a model with 79.2\% accuracy with DP parameter \changebars{$\epsilon = 5.53$}{$\epsilon = 1.34$},
    relative to 79.3\%  in FedScale, after 
    480 rounds of training.

        \item \Sys's  \cpu and network overhead is low to moderate: for a 1.2M parameter model, devices spend  15~ms of \cpu and 3.12~KiBs of network transfers most of the time, and occasionally (with a probability of $10^{-5}$ in a round) 13.4 min. of processor
    time and 234~MiBs of network transfers.

    \item 
    \sys's techniques improve over Orchard by up to $2.3\cdot 10^5\times$. %

\end{myitemize}

\noindent \emph{Testbed.}
Our testbed has machines of type 
\texttt{c5.24xlarge} on Amazon
EC2. Each machine has 
    96v\cpus, 192 GiB RAM, and
25 Gbps network bandwidth. 
We use a single machine for running \sys's server.
Meanwhile, 
we co-locate multiple devices on a machine: each device
    is assigned six \cpus given that modern mobiles
    have processors with four to eight \cpus.

\emph{Default system configuration.} 
Unless specified otherwise, we configure the systems to assume
$W=10^9$ total devices. 
For \sys, we set the default 
device sampling probability $q$ in DP-FedAvg to $10^{-5}$; i.e.,
the expected number of devices that contribute updates in a round is
$10^{4}$. We also configure \sys to use
ten decryption committees, where 
each committee has a total of $C = 45$ devices of which $A = 18$ 
may be malicious. %
The first decryption committee also serves as the 
master committee. We configure the DP-noise committee 
with $(A, C) = (40, 280)$. For Orchard, we configure
its committee to have 40 devices of which 16 may be malicious.\footnote{\sys's committees are larger because it must ensure, using a union bound, that
the chance of sampling more than $A$ malicious devices
 across \emph{any} of its committees is the same as 
in Orchard.}

\subsection{Comparison with FedScale}
\label{s:eval:FedScale}

\emparagraph{Accuracy.}
We evaluate several datasets and models to compare \sys with FedScale, specifically, the FedAvg algorithm
in FedScale. Figure~\ref{f:acc_models}
shows these datasets and models. \changebars{We use CNND and CNNF for two 
different CNN models: one dropout model~\cite{reddi2020adaptive} and the other from the
FedAvg paper~\cite{mcmahan2017communication}}{We use CNND and CNNF to represent two 
different CNN models, where the former is a dropout model~\cite{reddi2020adaptive} and the latter is from the 
FedAvg paper~\cite{mcmahan2017communication}}.

\Sys's accuracy depends on the DP parameters $\epsilon$ and $\delta$.
For Figure~\ref{f:acc_models} experiments, we set \changebars{$\epsilon = 5.04$}{$\epsilon = 1.04$} and 
$\delta = 1/W^{1.1}$. 
Further, for both systems, we set all other training parameters (batch size, 
the number of device-side training epochs, etc.) 
per the examples provided by FedScale for each dataset.

Figure~\ref{f:acc_models} compares the accuracies after 480 rounds of 
training (these models converge in roughly 400-500 rounds). Generally, \sys's accuracy
loss grows with the number of model parameters. The reason is that 
DP-FedAvg adds noise for every parameter and thus the norm of the noise increases
with the number of parameters.

Although \sys's accuracy loss is (very) high for a larger number of parameters, 
this loss is recoverable by increasing $\epsilon$ (but still keeping
it at a recommended value). 
Figure~\ref{f:acc_vs_rounds} shows accuracy for two  values of $\epsilon$ for two example
models. Increasing $\epsilon$ from \changebars{5.04 to 5.53}{1.04 to 1.34}
recovers the accuracy loss. 
For instance, for the CNNF model, FedScale's accuracy is
79.3\% after 480 rounds, while \sys's is 
79.2\%.
The reason is that as $\epsilon$
increases, more devices can contribute updates ($q$ increases), which increases
the signal relative to the differential privacy noise. Overall, \sys can give
competitive accuracy as plain federated learning for models with parameters
ranging from tens of thousand to a few million.

\begin{figure}[t]
    \footnotesize
    \centering
    
    \begin{tabular}{
        @{}
        *{1}{>{\raggedright\arraybackslash}b{.04\textwidth}}%
        *{1}{>{\raggedleft\arraybackslash}b{.04\textwidth}}  %
        *{1}{>{\raggedleft\arraybackslash}b{.06\textwidth}}  %
        *{1}{>{\raggedleft\arraybackslash}b{.04\textwidth}}  %
        *{1}{>{\raggedleft\arraybackslash}b{.06\textwidth}}  %
        *{1}{>{\raggedleft\arraybackslash}b{.04\textwidth}}  %
        @{}
        }
        &&\multicolumn{2}{c}{\cpu (ms)} &\multicolumn{2}{c}{network (KiB)}\\
        \hline
        Model & Size & FedScale & \sys & FedScale  & \sys \\
        \bottomrule

        LeNet & 49K  & 3.36E-4 & 2 & 3.93E-5  & 0.96  \\
        \hline
        CNND & 1.2M  & 9.50E-4 & 55 & 9.66E-4  & 3.87  \\
        \hline
        CNNF & 1.7M& 9.49E-4 & 77 & 1.35E-3  & 5.15  \\
        \hline
        AlexNet & 3.9M & 1.75E-3 & 170 & 3.12E-3  & 11.0  \\
        
        \bottomrule

    \end{tabular}
    \caption{Per device per round average cost for different models.}
    \label{fig:cost_vs_models}
    \label{f:cost_vs_models}
\end{figure}

\begin{figure*}[t]
\begin{subfigure}[t]{0.33\textwidth}{
\centering
\includegraphics[width=2.25in]{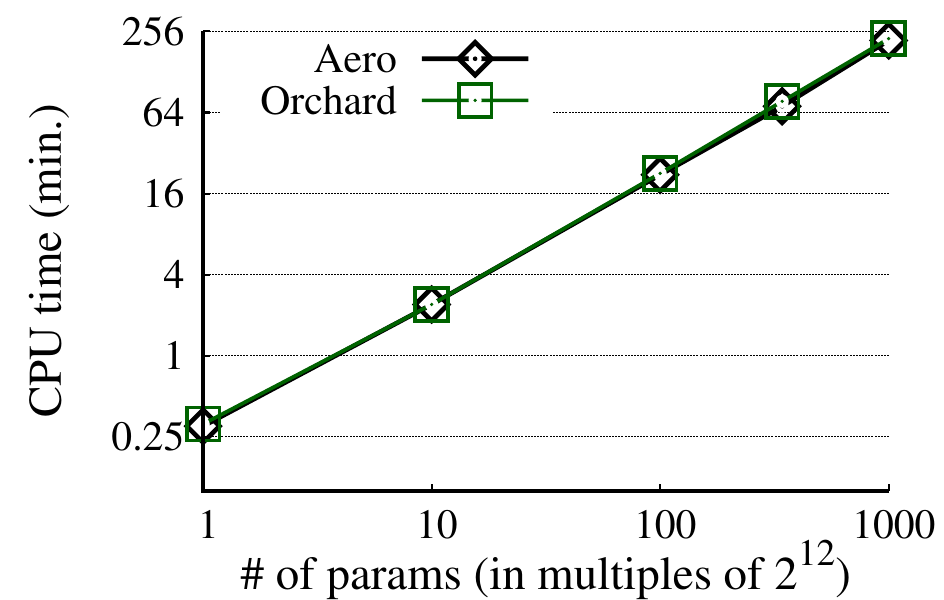}
\caption{Generator}
\label{f:generator_cpu}
}
\end{subfigure} 
\begin{subfigure}[t]{.33\textwidth}{
\centering
\includegraphics[width=2.25in]{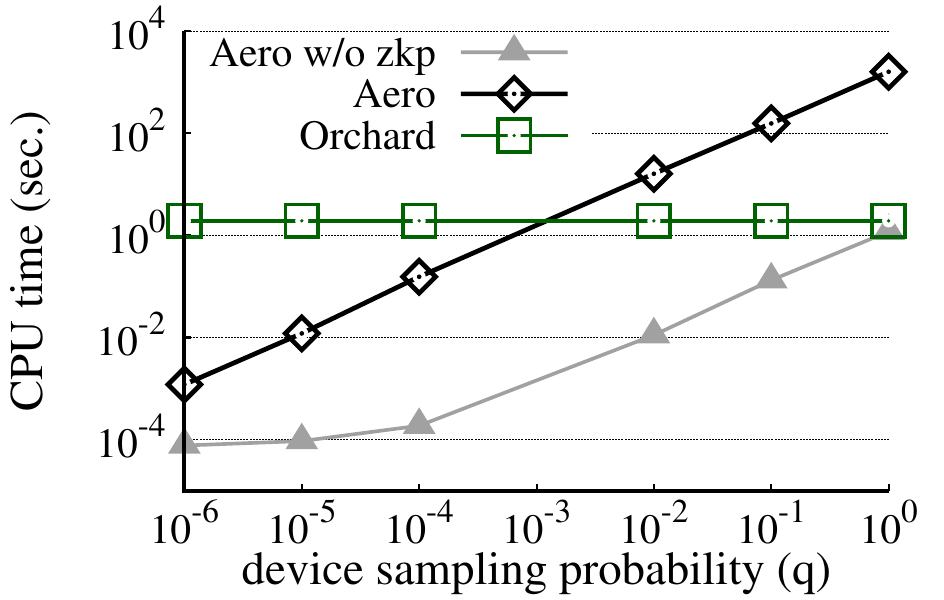}
\caption{Verifier}
\label{f:verifer_cpu}
} 
\end{subfigure}
\begin{subfigure}[t]{.33\textwidth}{
\centering
\includegraphics[width=2.25in]{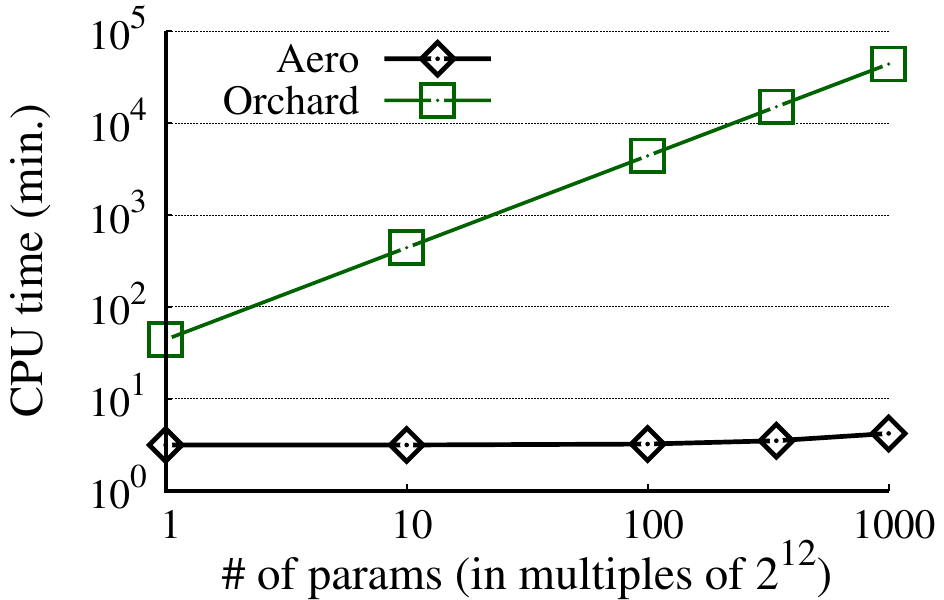}
\caption{Decryption committee}
\label{f:committee_cpu}
}
\end{subfigure}
\caption{%
\cpu time per device per round of training
    for different device roles in \sys and Orchard.
}
\label{f:eval3-micro-cpu}
\label{fig:eval3-micro-cpu}
\end{figure*}

\begin{figure*}[ht]
\centering
\begin{subfigure}[t]{0.33\textwidth}{
\centering
\includegraphics[width=2.25in]{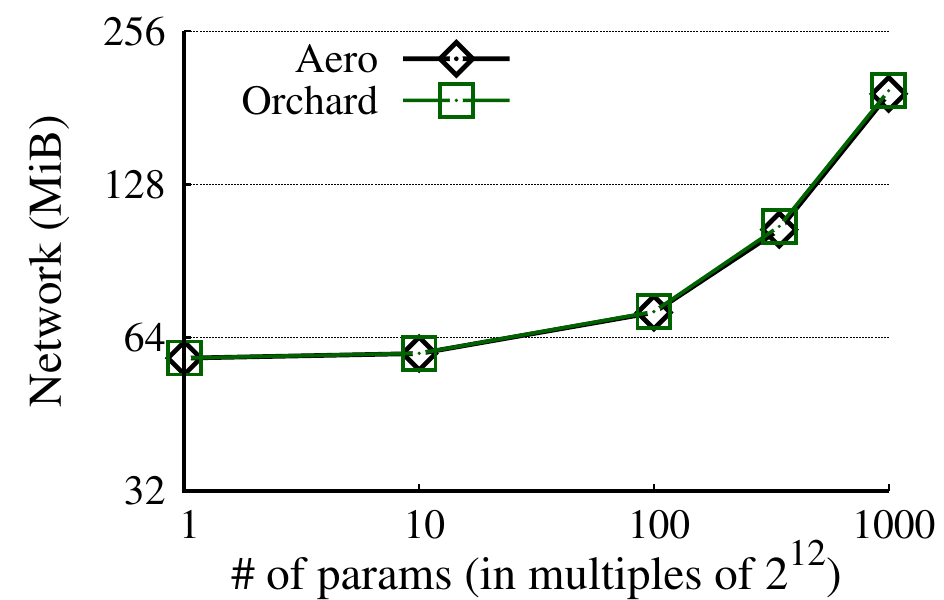}
\caption{Generator}
\label{f:generator_net}
}
\end{subfigure} %
\begin{subfigure}[t]{.33\textwidth}{
\centering
\includegraphics[width=2.25in]{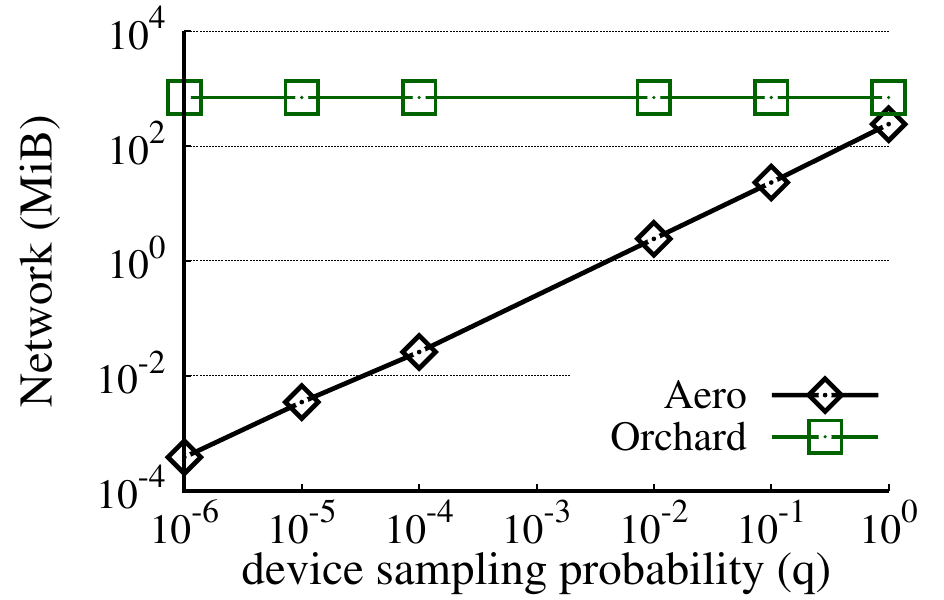}
\caption{Verifier}
\label{f:verifier_net}
} %
\end{subfigure}
\begin{subfigure}[t]{.33\textwidth}{
\centering
\includegraphics[width=2.25in]{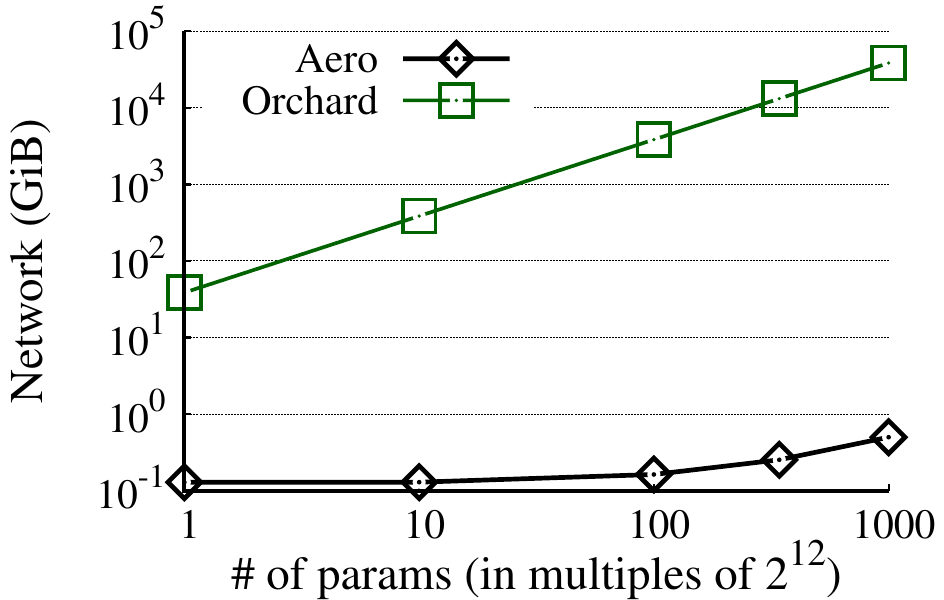}
\caption{Decryption committee}
\label{f:committee_net}
}
\end{subfigure}
\caption{%
Network transfers per device per round of training
    for different device roles in \sys and Orchard.
    }
\label{f:eval3-micro-net}
\label{fig:eval3-micro-net}
\end{figure*}

\emparagraph{Device overhead.}
Another cost of privacy relative to plain federated learning is increased
device overhead. Figure~\ref{f:cost_vs_models} summarizes the average \cpu and network cost per round
per device for the four models on the FEMNIST dataset. (We picked the FEMNIST dataset just as an example, but the results for the other datasets are qualitatively the same.)

Overall, an \sys device on average (considering 
the different types of \sys devices) spends $5.9\cdot 10^3-9.7\cdot 10^4\times$ higher \cpu
and $3.5\cdot10^3-2.4\cdot10^4\times$ higher network relative to FedScale.
This overhead
is due to the fact that FedScale does not use any cryptographic operations, while \sys devices
use many, for example, encryption and ZK-proofs during the generate phase, and 
verifiable aggregation during the add phase. However, \sys's overhead, at least, 
on average, is low (Figure~\ref{f:cost_vs_models}). Further, as we will show next, %
\sys's worst-case overhead is also moderate.

\subsection{Comparison to Orchard}
\label{s:eval:Orchard}
Both \sys and Orchard have multiple types of devices. 
\Sys has devices that participate in the master committee,
generate updates (or Gaussian noise),
verify the aggregator's work, and  
participate in the decryption committee. Similarly, Orchard
has generator, verifier, and committee devices.
We compare overhead for these devices separately.

\emparagraph{Generator device overhead.}
The overhead for the generators
changes only with the model size (after excluding the training time to 
generate the plain updates).
Thus, we vary the number
    of model parameters and 
    report overhead. 

Figure~\ref{f:generator_cpu} shows the \cpu time 
and Figure~\ref{f:generator_net} shows the network transfers
with a varying number of model parameters. 
These overheads grow linearly with the number of model parameters (the network
overhead is not a straight line as it includes a fixed cost of 60~MiB to download ZK-proof proving keys).%
The reason is that the dominant operations for a generator device
are generating ZK-proofs and shipping 
ciphertexts to the aggregator. The number of both operations is proportional to the number of parameters (\S\ref{s:design:generate}). %

In terms of absolute overhead, 
a specific data point of interest is a million-parameter model, e.g., the
CNND model with 1.2M parameters.
    For this size,
    a generator device spends 1.01 hours in \cpu time, or 
    equivalently 13.4 minutes of latency (wall-clock time) over six cores.
The generator also sends 101~MiB of data over the network. 
These overheads are moderate, considering the fact that the probability that a device will be a generator in a round is small: $10^{-5}$.

Finally, the \cpu and network overhead for \sys and Orchard 
    is roughly the same.%
The reason is that the dominant operations for the two systems
are common: ZK-proofs and upload of ciphertexts.

\emparagraph{Verifier device overhead.}
Figure~\ref{f:verifer_cpu} shows \cpu
and Figure~\ref{f:verifier_net} shows network
overhead for the verifier devices 
that participate in the verifiable aggregation protocol (\S\ref{s:design:add}). 
These experiments
fix the number of model parameters to 1.2M and 
vary the probability $q$ with which a verifier device samples
summation trees to inspect (recall that a verifier device
in \sys checks $q \cdot \ell$ summation trees). For Orchard,
overhead does not change with $q$ ($q$ is effectively 1).

Overall, \sys's verifier devices, which are the bulk of the devices in 
the system, are efficient consuming a few milliseconds 
of \cpu and a few KiBs of network transfers. 
    For instance, for $q = 10^{-5}$, \sys incurs 3.12~KiB in network and 15~ms of \cpu time, while Orchard incurs 1.96 seconds (130$\times$)  and 738~MiB ($2.36\cdot 10^5\times$).

Comparing \sys with Orchard, a verifier in \sys 
    consumes lower \cpu than Orchard for smaller values of 
    $q$ but a  higher \cpu for larger $q$.
This trend is due to constants: 
    even though an \sys device
    checks $q \cdot \ell$ summation trees and $3s$ ciphertexts
    in each tree
    versus 
    $\ell \cdot 3s$ ciphertexts in Orchard,
    \sys devices verify the ZK-proofs to address the 
    reuse-of-keys issue,
    while Orchard 
    does not have such a requirement (\S\ref{s:design:add}).
Each proof check 
    takes $\approx$700~ms on a single \cpu of 
    \texttt{c5.24xlarge}. 
Indeed, \sys w/o ZK-proof check (another line in the plot)
    is strictly better than Orchard.

\Sys's network
overhead
    increases linearly with $q$, while
Orchard's stays constant as it does not do sampling (Figure~\ref{f:verifier_net}).
Notably, when $q = 1$, i.e., when
    \sys and Orchard check the same number of ciphertexts,
    a \sys verifier consumes 
    251 MB, which is $\approx1/3$rd of Orchard.
This is \changebars{because polynomial identity testing}{due to the use of
the polynomial identity testing which} allows a
    \sys verifier to download evaluations of ciphertext polynomials rather
    than the full polynomials from non-leaf vertices 
    \changebars{}{of the summation trees }(\S\ref{s:design:add}).

\emparagraph{Committee device overhead.}
Figure~\ref{f:committee_cpu} and Figure~\ref{f:committee_net} 
show the \cpu and network overhead of decryption committee devices as a function of 
the model size. (In \sys, the first decryption committee also serves as the master committee.)

\Sys's overheads are much lower than Orchard's---for 1.2M parameters, %
    \cpu time
    is 206~s in \sys versus 214 hours in Orchard (i.e., $3751\times$ lower),
    and network 
    is 234~MiB in \sys versus 11~TiB in Orchard (i.e., $4.8 \cdot 10^{4}\times$ lower).
This improvement is for two reasons. 
    First, \sys divides the decryption of multiple ciphertexts across
committees, and thus each performs less work.
Second, \sys uses the distributed decryption protocol
(\S\ref{s:design:release}), while Orchard
    uses the general-purpose SCALE-MAMBA MPC~\cite{scaleMamba}.

\section{Related work}
\label{s:relwork}

\Sys's goal is to add the rigorous guarantee of 
differential privacy to federated learning---at low device overhead. 
This section compares \sys to prior work with similar goals.

\emparagraph{Local differential privacy (LDP).}  In LDP,
devices \emph{locally} add noise to their updates
    before submitting them
    for aggregation~\cite{duchi2013local, erlingsson2014rappor, ijcai2021-217, he2020secure,
pathak2010multiparty, truex2020ldp,seif2020wireless,
bhowmick2018protection, hao2019towards, sun2020federated, grafberger2021fedless, 
nguyen2016collecting, wang2019collecting, niu2019secure, lu2019blockchain, chen2018machine, mugunthan2020blockflow,chen2020practical,ding2021differentially}.
On the plus side, the privacy guarantee in LDP
    does not depend on the behavior of the aggregator, as devices add noises locally.
Further, LDP is scalable as it adds small device-side overhead 
    relative to 
    plain federated learning.
However, on the negative side, since each device perturbs its update, the trained model
can have a large error.

\emparagraph{Central differential privacy (CDP).}
Given the accuracy loss in LDP, many systems
    target CDP~\cite{froelicher2017unlynx,
    zeng2022aggregating,
    sebert2022protecting,
truex2019hybrid,
xu2019hybridalpha,
chase2017private,
rastogi2010differentially,
stevens2021efficient,
hynes2018efficient,
roth2019honeycrisp,
roth2020orchard,
xu2022detrust}.
The core challenge is of hiding
     sensitive device updates from the aggregator. 

Several systems in this category
 target a setting
of a few tens of devices to a few 
thousand devices~\cite{xu2022detrust,xu2019hybridalpha, truex2019hybrid, stevens2021efficient, sebert2022protecting}. 
These systems require \emph{all} 
devices to participate in one or more cryptographic primitives, and
thus their overhead grows with the number of devices. 
For example, in secure aggregation based FLDP~\cite{stevens2021efficient}, 
    each device generates a secret key, 
then masks its update using the key, before
sending the masked update to the aggregator. Then,
the devices 
securely sum  their masks to subtract them from 
the aggregator's result. This latter protocol 
    requires each device to secretly share its mask with all others. %

Chase et al.~\cite{chase2017private} do not require their protocol
to scale with the number of devices: two
devices aggregate updates from all others
before generating and adding DP noise via Yao's MPC protocol~\cite{yao1982protocols}.
The issue is that if the adversary
compromises the two devices, it learns the updates.

Honeycrisp~\cite{roth2019honeycrisp}, Orchard~\cite{roth2020orchard}, 
    and Mycelium~\cite{roth2021mycelium} target a setting
of a billion devices. One of their key insights is to run 
    expensive cryptographic protocols among
a small, randomly-sampled committee, while leveraging 
an untrusted resourceful aggregator to help with the aggregation.
Among the three systems, Orchard supports learning tasks, while Honeycrisp
supports aggregate statistics and Mycelium supports graph analytics. The
limitation of Orchard is that it imposes a large overhead on the devices 
    (\S\ref{s:problem:solutions}, \S\ref{s:eval}). 
    \Sys improves over Orchard by several orders
of magnitude (\S\ref{s:eval}).

An alternative to cryptography is to use trusted hardware, e.g.,
Intel SGX~\cite{hynes2018efficient}. These systems add negligible overhead over plain federated learning, 
    but trusting the hardware design and  
    manufacturer is a strong assumption~\cite{fei2021security,TrustZoneAttacks,ArmSEVAttacks}.

\emparagraph{No differential privacy.}
Many systems 
    provide a weaker notion of 
    privacy 
    than differential privacy, for functionality such as 
    federated machine
learning~\cite{aono2017privacy,rathee2022elsa, dong2020eastfly, fu2020vfl, jiang2020federated,
jiang2021flashe,
liu2019secure, ma2021privacy, mandal2019privfl, 254465, sav2020poseidon, xu2022hercules,
beguier2020efficient, chen2021ppt, chowdhury2021eiffel, ergun2021sparsified,
fereidooni2021safelearn, guo2020secure, hao2021efficient, kadhe2020fastsecagg, 
li2021secure, liu2020boosting, so2021turbo, xu2019verifynet, 
zhang2021dubhe, mo2021ppfl, hashemi2021byzantine, quoc2021secfl,sav2022privacy},
statistics~\cite{corrigan2017prio}, and
    aggregation~\cite{bell2020secure,bonawitz2017practical,liu2022dhsa,wan2022information,liu2022efficient}.
For instance, BatchCrypt~\cite{254465} 
    uses Paillier AHE~\cite{damgaard2001generalisation} 
    to hide updates from the aggregator. 
The promise is that
    the adversary learns only the aggregate of the data of many devices.
The fundamental issue is that aggregation does not provide
    a rigorous guarantee: one can learn individual training data from the trained model parameters~\cite{zhu2019deep,melis2019exploiting,briland2017deep,shokri2017membership}.

\section{Summary}
\label{s:summary}

Federated learning over a large number of mobile devices
is getting significant attention both in industry and academia.
One big challenge of current practical systems, those that provide good accuracy and efficiency,
is the trust
they require: the data analyst must say ``let's trust that the 
server will not be compromised''.
\Sys adds an alternative. It shows that one can perform 
FL with good accuracy, moderate overhead, and the rigorous  guarantee
of differential privacy without trusting a central server
or the data analyst. \Sys improves the trade-off
by focusing on a specific type of learning algorithms and tuning
system architecture and design to these algorithms (\S\ref{s:design}).
The main evaluation highlight is that \sys has comparable accuracy to plain
federated learning, and improves over prior work Orchard
that has strong guarantees by five orders of magnitude (\S\ref{s:eval}).

\frenchspacing

\begin{flushleft}
\setlength{\parskip}{0pt}
\setlength{\itemsep}{0pt}
\bibliographystyle{abbrv}
\bibliography{conferences-long-with-abbr2,paper}
\end{flushleft}

\appendix
\raggedend
\clearpage
\section{Supplementary material}
\newtheorem{claim}{Claim}[section]
\newtheorem{assumption}{Assumption}[section]
\subsection{Privacy proof}
\label{s:appendix:privacy_proof}
The goal of \sys is to provide differential privacy for a class of federated learning algorithms. We will take DP-FedAvg as an example and prove that \sys indeed meets its goal in multiple steps. The proof for other algorithms such as DP-FedSGD is similar. The outline of the proof is as follows.

First, we will introduce a slightly modified version of DP-FedAvg
that makes explicit the behavior of the malicious aggregator and devices. 
This modified version changes line~\ref{l:dpfedavg:sampleclients} and line~\ref{l:dpfedavg:aggregateupdates}
in Figure~\ref{f:dpfedavg}.
For instance, we will modify line~\ref{l:dpfedavg:sampleclients} in Figure~\ref{f:dpfedavg}
to show that a byzantine aggregator may allow malicious devices
to be sampled in a round. Appendix~\ref{sssec:dp} introduces the modified version and 
shows that the changes do not impact DP-FedAvg's differential privacy guarantee.

Next, we will prove that \sys executes the modified
DP-FedAvg algorithm faithfully, by showing that enough Gaussian noise will be added (Appendix~\ref{sssec:noise}) and if an adversary introduces an
error into the aggregation, it will be caught with high probability (Appendix~\ref{sssec:agg}).

Finally, we will prove that after aggregation, \sys's
decryption protocol does not leak any information beyond the allowed output
of DP-FedAvg (Appendix~\ref{sssec:decrypt}).

We will not cover security of the protocols used in the setup phase, e.g., the sortition protocol
to form committees, because \sys does not 
innovate on these protocols. With the above proof structure, we will show \sys provides the differential privacy guarantee to 
honest devices' data.

Before proceeding to the proof, we introduce a few definitions. In the main body of the paper, for
simplicity, we did not distinguish between honest-but-offline and malicious
devices for the DP-noise committee (\S\ref{s:design:generate}). Instead, we considered all offline devices as
malicious for this committee, since it's not possible to tell whether an offline device is malicious or not.

However, for devices that generate updates, we do protect honest-but-offline devices' data. So when
talking about generator devices, we use the following terms: 
\begin{itemize}
    \item honest-and-online devices,
    \item honest-but-offline devices,
    \item honest devices: including both honest-and-online and honest-but-offline devices
    \item malicious devices %
\end{itemize}
As for the DP-noise committee members, we still follow the previous definition, namely
\begin{itemize}
    \item honest members: honest and online members
    \item malicious members: malicious or honest-but-offline members
\end{itemize}

\subsubsection{DP-FedAvg}
\label{sssec:dp}
As mentioned above, \sys must take into account the behavior of malicious entities for
the computation in line~\ref{l:dpfedavg:sampleclients} and line~\ref{l:dpfedavg:aggregateupdates} of 
Figure~\ref{fig:dpfedavg}.

There are two reasons, \sys must modify line~\ref{l:dpfedavg:sampleclients}: first, a
malicious aggregator may filter out honest-but-offline devices' data in the
aggregation; second, a malicious aggregator may add malicious devices' data in
the aggregation.

To capture the power of a malicious aggregator, we change line~\ref{l:dpfedavg:sampleclients} to be

\textit{$C^t\leftarrow$ subset of (sampled honest users with probability q) + (some
malicious devices)}

We must modify line~\ref{l:dpfedavg:aggregateupdates} because the DP-noise committee is
likely to add more noise as noted in the design of the generate phase (\S\ref{s:design:generate}). 
To capture this additional noise, we
change line~\ref{l:dpfedavg:aggregateupdates} to be

\textit{$\Delta^t\leftarrow \sum_k \Delta_k^t + \mathcal{N}(0,I\sigma^2)$ + some
additional bounded noise}

In the remainder of this section, we will prove that the modifications
preserve the privacy guarantee of the original
DP-FedAvg algorithm~\cite{brendan2018learning}.

\emparagraph{Device sampling (line~\ref{l:dpfedavg:sampleclients}).} For device sampling, there are two possible attacks:
either some malicious devices' data will be included or some honest devices'
data will be filtered out.

The first case is not a problem, since it's equivalent to post-processing: for
example, after aggregation, the malicious aggregator can add  malicious
updates. Post-processing of a deferentially private result does not 
affect the differential privacy guarantee (this follows from the post-processing lemma
in the differential privacy literature~\cite{zhu2021bias,dwork2006calibrating}).

To prove that filtering out honest devices' data will not impact the privacy guarantee, 
we'll first give intuition and then a
rigorous proof. Intuitively, a larger sampling probability (larger $q$) leads to more privacy
loss, because a device's data is more likely to be used in training. So
informally, if each device is expected to contribute an update fewer times, the
privacy loss is expected to be less. Now coming back to the case where the aggregator
filters out honest devices' data, it is obvious that each device is expected to
contribute updates no more frequently than without filtering, which means the
privacy loss is expected to be no more than without filtering.

In the original paper of DP-FedAvg~\cite{brendan2018learning}, the DP guarantee relies on the moments accountant introduced by Abadi et al.~\cite{abadi2016deep}, whose tail
bound can be converted to $(\epsilon, \delta)$-differential
privacy. To be precise, the proof uses a lemma (which we will introduce shortly) that
gives the moments bound on Gaussian noise with random sampling, which is
equivalent to $(\epsilon, \delta)$-differential privacy. We cite this lemma as Lemma~\ref{lm:moment} in the appendix.

So next we will show by replacing the original random sampling with random
sampling plus filtering, the moment bound still holds.
In the discussion, without losing generality, we will focus on functions whose sensitivity is 1, for example, 
the $\textsc{UserUpdate}$ function that 
does local training and gradient clipping with $S=1$ in DP-FedAvg (Figure~\ref{f:dpfedavg}).
Notice that for any function $f'$ whose sensitivity is $S' \neq 1$, we can always construct a 
function $f=f'/S'$ whose sensitivity is 1.

\begin{lemma}\label{lm:moment}
Given any function $f$, whose norm $\|f(\cdot)\|_2 \le 1$, let
$z\geq 1$ be some noise scale and $\sigma=z\cdot \|f(\cdot)\|_2$, let $d=\{d_1,...,d_n\}$ be a database, 
let $\mathcal{J}$ be a sample from $[n]$ where each $i\in [n]$ is
chosen independently with probability $q\le \frac{1}{16\sigma}$, then for any
positive integer $\lambda\le \sigma^2 \ln\frac{1}{q\sigma}$, the function
$\mathcal{G}(d)=\sum_{i\in \mathcal{J}} f(d_i) + \mathcal{N}(0,\sigma^2 \boldsymbol{I})$
satisfies $$\alpha_{\mathcal{G}}(\lambda)\leq
\frac{q^2\lambda(\lambda+1)}{(1-q)\sigma^2}+O(q^3\lambda^2/\sigma^3).$$
\end{lemma}

Notice that if the bound on $\alpha_{\mathcal{G}}$ does not change when some data points are filtered out after sampling, the differential privacy guarantee will not change either. We refer readers to \cite{abadi2016deep} for more details.

\begin{claim}
Let $\mathcal{J}$ be a sample from $[n]$ where each $i\in [n]$ is chosen independently
with probability $q\le \frac{1}{16\sigma}$. Let $\mathcal{J}'\subseteq \mathcal{J}$ be some arbitrary subset of $\mathcal{J}$.
Then for any positive integer
$\lambda\le \sigma^2 \ln\frac{1}{q\sigma}$, the function
$\mathcal{G'}(d)=\sum_{i\in \mathcal{J}'} f(d_i) + \mathcal{N}(0,\sigma^2 \boldsymbol{I})$
also satisfies $$\alpha_{\mathcal{G'}}(\lambda)\leq
\frac{q^2\lambda(\lambda+1)}{(1-q)\sigma^2}+O(q^3\lambda^2/\sigma^3).$$
\end{claim}

\begin{proof}
Let's abuse the notation a little bit and define $\mathcal{G}_{f,\mathcal{J}}$ to be
$$\mathcal{G}_{f,\mathcal{J}}(d)=\sum_{i\in \mathcal{J}} f(d_i) + \mathcal{N}(0,\sigma^2 \boldsymbol{I}).$$

To prove this claim, since by definition $\mathcal{G'}(d)= \mathcal{G}_{f,\mathcal{J'}}(d)$, we just need to show 
$$\alpha_{\mathcal{G}_{f,\mathcal{J'}}}(\lambda)\leq
\frac{q^2\lambda(\lambda+1)}{(1-q)\sigma^2}+O(q^3\lambda^2/\sigma^3).$$

To do this, suppose for now we have a $f'$ on $\mathcal{J}$ whose sensitivity is no greater than 1 and gives the same output as $f$ on $\mathcal{J}'$, namely
$$\|f'(\cdot)\|_2 \le 1,\mathcal{G}_{f',\mathcal{J}}(d)=\mathcal{G}_{f,\mathcal{J'}}(d).$$

By applying Lemma~\ref{lm:moment} on $\mathcal{G}_{f',\mathcal{J}}$, we get
$$\alpha_{\mathcal{G}_{f',\mathcal{J}}}(\lambda)\leq
\frac{q^2\lambda(\lambda+1)}{(1-q)\sigma^2}+O(q^3\lambda^2/\sigma^3).$$

Since $\mathcal{G}_{f',\mathcal{J}}(d)=\mathcal{G}_{f,\mathcal{J'}}(d)$,
$$\alpha_{\mathcal{G}_{f',\mathcal{J}}}(\lambda)=\alpha_{\mathcal{G}_{f,\mathcal{J'}}}(\lambda).$$

Combining the preceding two equations, we get
$$\alpha_{\mathcal{G}_{f,\mathcal{J'}}}(\lambda)=\alpha_{\mathcal{G}_{f',\mathcal{J}}}(\lambda)\leq
\frac{q^2\lambda(\lambda+1)}{(1-q)\sigma^2}+O(q^3\lambda^2/\sigma^3).$$

The remaining task is to construct such a $f'$ on $\mathcal{J}$, where $\|f'(\cdot)\|_2\le 1$, to give the same output as $f$ on $\mathcal{J}'$. One possible $f'$ is as follows:
$$f'(d_i)= \begin{cases}
       f(d_i) &\quad\text{if } i\in \mathcal{J}',\\
       0 &\quad\text{if } i\in \mathcal{J}-\mathcal{J}'.
     \end{cases}$$
It's easy to prove $\sum_{i\in \mathcal{J}'} f(d_i)=\sum_{i\in \mathcal{J}} f'(d_i)$.

Next, for the sensitivity of $f'$, it is not difficult to see
$\|f'(\cdot)\|_2\le \|f(\cdot)\|_2$, since by removing or adding one entry to the
database, $f'$ will incur either the same change as $f$ or no change.
\end{proof}
So far we have  proved that if only a subset of 
selected devices are included in aggregation or more malicious devices' 
data is included, the differential privacy guarantee
will not be impacted.

\emparagraph{Gaussian Noise (line~\ref{l:dpfedavg:aggregateupdates}).}
Recall that we also add some additional noise in line~\ref{l:dpfedavg:aggregateupdates} in Figure~\ref{f:dpfedavg}.
For the Gaussian noise, we need to prove additional noise will not 
impact privacy, which is not difficult to show, 
since this change is also equivalent to post-processing. 

With the above proof, we've showed that our modified DP-FedAvg 
provides the same privacy guarantee 
as the original DP-FedAvg.

\subsubsection{DP-noise committee}\label{sssec:noise}
This section shows that the $\mathcal{N}(0,I\sigma^2)$ part of 
line~\ref{l:dpfedavg:aggregateupdates} of our modified DP-FedAvg is executed faithfully.

We have already covered in the generate
phase (\S\ref{s:design:generate}) that $\mathcal{N}(0,I\sigma^2)$ amount of noise
will be generated by the DP-noise committee as long as the DP-noise committee does not violate
its threshold: less than $A$ out of $C$ devices of the DP-noise committee are indeed malicious.
Thus, in this section we will derive the probability of a committee having fewer than some threshold of
honest members. We will follow the same way to compute the probability as in
Honeycrisp~\cite{roth2019honeycrisp}, which is a 
building block for Orchard~\cite{roth2020orchard}.

\begin{claim}
(\sys) If a randomly sampled DP-noise committee size is $C$, the probability of a committee member being malicious is $f$, the probability that
committee has fewer than $(1-t)\cdot C$ honest members is upper-bounded
by $\mathit{p}=e^{-fC}(\frac{ef}{t})^{tC}$, when $1>t\geq f.$
\end{claim}
\begin{proof}
We treat each member being malicious as independent events. Let $X_i$ be a random variable
$$X_i=\begin{cases}
1, &\quad\text{if member i is malicious,}\\
0, &\quad\text{if member i is honest.}
\end{cases}$$Let $X=\sum_i X_i$ be the random variable representing the number
of malicious members. If $t\geq f$, the Chernoff bound shows that $$Pr(X \geq tC) \leq e^{-fC}(\frac{ef}{t})^{tC}.$$
So the probability of fewer than $(1-t)\cdot C$ members being honest will be upper-bounded by $\mathit{p}.$
\end{proof}

With (high) probability, the lower bound of honest committee members
will be $(1-t)C$. As long as each honest member contributes
$\frac{1}{(1-t)C}\cdot \mathcal{N}(0,\sigma^2\boldsymbol{I})$ noise, the total amount
of noise will be no less than $\mathcal{N}(0,\sigma^2\boldsymbol{I})$.

To give some examples of what committee sizes could be, 
when $f=0.03$, we may set $t=1/7, C=280$, to achieve $\mathit{p}=4.1\cdot
10^{-14}$; when $f=0.05$, we may set $t=1/8, C=350$ to achieve $\mathit{p}=9.78\cdot
10^{-7}$ or $C=450$ to achieve $\mathit{p}=1.87\cdot 10^{-8}$; and, 
when $f=0.10$, we may set $t=1/5, C=350$ to achieve $\mathit{p}=1.34\cdot
10^{-6}$ or $C=450$ to achieve $\mathit{p}=2.82\cdot 10^{-8}$.

\subsubsection{Aggregation}\label{sssec:agg}
This section will prove that the additions in line~\ref{l:dpfedavg:aggregateupdates} in our modified DP-FedAvg 
are executed faithfully in \sys. Otherwise, the aggregator will be caught with a high probability.

We will first prove the integrity of additions as in \sys's add phase protocol and in Honeycrisp. Next, we will
prove the freshness guarantee that no ciphertexts from previous rounds can be included in the current aggregation. Finally, we
will prove how these two proofs together  show that line~\ref{l:dpfedavg:aggregateupdates} of 
modified DP-FedAvg is executed faithfully by \sys.

\emparagraph{Integrity.} We will start with the integrity claim from
Honeycrisp~\cite{roth2019honeycrisp}. 

\begin{claim}\label{sssec:aggclaim}
At the end of add phase, if no device has found malicious activity by the
aggregator $\mathcal{A}$, the sum of the ciphertexts published by $\mathcal{A}$
is correct (with high probability) and no inputs of malicious nodes are
dependent on inputs of honest nodes (in the same round).  
\end{claim}

\begin{proof}
We refer readers to Honeycrisp~\cite{roth2019honeycrisp} for more details. Here
we will just give a short version for demonstration.

We will first prove no inputs of malicious devices are dependent on inputs of
honest devices in the same round. Next we will prove if a malicious aggregator
introduces an error into a summation tree, it will be caught with high
probability.

First, assume for the sake of contradiction that it is possible for a malicious
device to set its ciphertext to be $c$ that is from an honest device in the
current round. The adversary needs to produce a $t=Hash(r||c||\pi)$ and include
$t$ in the Merkle tree $MC$, before the honest device reveals $c$. Under the Random Oracle
assumption in cryptography, this is not possible. 
So no inputs of malicious devices are
dependent on inputs of honest devices in the same round.

Next we need to show if $\mathcal{A}$ introduces an error into a summation tree $ST$, it will be
caught with high probability. In particular, here we will just show the case
where $\mathcal{A}$ introduces an error into one of the leaf nodes. Similar analysis can be
done for non-leaf nodes, which is presented in
Honeycrisp~\cite{roth2019honeycrisp}.

Suppose the total number of verifiers is $W$ and the total number of 
leaf nodes in one summation tree is $M'$ (Figure~\ref{f:addphase}). 
Here we make the assumption that $M'\approx qW$, where $q$ is the sampling
probability. We will prove this assumption later in this section.

Suppose $\mathcal{A}$ introduces an error into a particular leaf node $j\in[0,M'-1]$.
The probability of an honest device picking any $v_{init}$ to have
$j\in[v_{init},v_{init}+s]$ is
$$\frac{qs}{M'}.$$
Since there are at least $(1-f)W$ honest devices, the probability
of no honest device checking $j$ is 
$$(1-\frac{qs}{M'})^{(1-f)W}
=(1-\frac{s}{W})^{(1-f)W}
\leq e^{-(1-f)s}.$$

Similarly we can prove if $\mathcal{A}$ introduces error into non-leaf nodes, it will be caught with a high probability.
For instance, if $f=3\%,s=5$, the probability is about 0.007; if $f=3\%, s=20$, the probability is about $10^{-11}$.

\end{proof}

As noted above, proof relies on an assumption about the threshold of Sybils (pseudonym leaf nodes)
the aggregator can introduce into the summation tree.  Intuitively, if the
aggregator can introduce as many Sybils as it wants to make $M'\gg qW$, then the probability of
each node being covered will decrease by a large factor, thus impacting privacy. 
For example, if $W=100,q=0.01$, $M'$ is expected to be close to $qW=1$ and the 
100 verifiers expect to verify only 1 leaf node. However, 
if the aggregator introduces 99 additional leaf nodes into the summation tree 
while the 100 verifiers still expect to verify only 1 leaf node, obviously many malicious nodes are very
likely to be missed by the verifiers. Before claiming in \sys there is a threshold of Sybils, we need to introduce an assumption from Honeycrisp~\cite{roth2019honeycrisp}, which we will also use in our claim.

\begin{assumption}\label{assumption:device_number}
(Honeycrisp) All devices know an upper bound $W_{max}$ and a lower bound
$W_{min}$ of the number of potential participating devices in the system.  If
the true number of devices is $W_{tot}$, then by definition: $W_{min} \leq W_{tot} \leq W_{max}$.
We assume $\frac{W_{max} - W_{tot}}{W_{min}}$ is always below some constant
(low) threshold (this determines the portion of Sybils the aggregator $\mathcal{A}$ could make
without getting caught). $\frac{W_{max}}{W_{min}}$ should also be below some
(more generous) constant threshold.  \end{assumption}

With the above assumption, similarly, we claim in \sys:

\begin{claim}
(\sys) There is a upper bound on the portion of
Sybils a malicious aggregator can introduce without being caught.
Precisely, all devices know an upper bound $M_{max}$ and a lower bound
$M_{min}$ of the number of potential leaf nodes in the summation tree. If the true number of leaf nodes is $M_{tot}$, then by definition: $M_{min} \leq M_{tot} \leq M_{max}$.
In \sys, $\frac{M_{max} - M_{tot}}{M_{min}}$ is always below some constant
(low) threshold (this determines the portion of Sybils the aggregator $\mathcal{A}$ could make
without getting caught). $\frac{M_{max}}{M_{min}}$ will also be below some
(more generous) constant threshold. 
\end{claim}

\begin{proof}
Note that in this claim, by "leaf nodes", we mean leaf nodes corresponding to a honest generator device. Once we know a upper bound $M_{max}$ and a lower bound $M_{min}$ of the number of leaf nodes, in the worst case, a malicious aggregator could introduce at most $M_{max}-M_{min}$ Sybils. If the true number of leaf nodes is $M_{tot}$, then the aggregator could introduce at most $M_{max}-M_{tot}$ Sybils.

Let's first consider $M_{min}$ and $M_{max}$.

Informally, if the number of devices $W$ is large enough, the number of leaf nodes (generators) is likely to be close to the expectation $qW$. Suppose for now there are two constants capturing this "closeness": $k_{min}$, $k_{max}$, where $k_{min}\approx 1, k_{max} \approx 1$. Then, we have
$$k_{min}\cdot qW_{min}\leq M_{min} \leq k_{max}\cdot qW_{min},$$
$$k_{min}\cdot qW_{max}\leq M_{max} \leq k_{max}\cdot qW_{max}.$$

Consider the fraction $\frac{M_{max}}{M_{min}}.$
$$\frac{M_{max}}{M_{min}}\leq \frac{k_{max}\cdot qW_{max}}{k_{min}\cdot qW_{min}}=\frac{k_{max}}{k_{min}}\cdot \frac{W_{max}}{W_{min}}$$

$\frac{W_{max}}{W_{min}}$ is below some constant threshold as in Assumption~\ref{assumption:device_number}. If $\frac{k_{max}}{k_{min}}$ is smaller than some constant threshold, then it's reasonable to assume like in Honeycrisp that $\frac{M_{max}}{M_{min}}$ is below some constant threshold. We will discuss the values of $k_{max}$ and $k_{min}$ at the end of the proof.

As for the fraction $\frac{M_{max}-M_{tot}}{M_{min}}$, it's not difficult to see that,
$$\frac{M_{max}-M_{tot}}{M_{min}}\leq \frac{k_{max}W_{max}-k_{min}W_{tot}}{k_{min}W_{min}}$$

Similarly if both $k_{max}$ and $k_{min}$ are
close to 1, this fraction will be close to $\frac{W_{max}-W_{tot}}{W_{min}}$, which is below some small constant (low) threshold as specified in Assumption~\ref{assumption:device_number}. So in \sys, it is reasonable to assume this fraction is also below some constant (low) threshold.

Lastly, let's discuss $k_{max}$ and $k_{min}$. Consider a general case where the total number of devices is $W$ and the sampling probability is $q$. The number of sampled devices (leaf nodes) is $X$. Chernoff bound states that,

$$Pr(X< k_{min}qW) <
(\frac{e^{k_{min}-1}}{{k_{min}}^{k_{min}}})^{qW}.$$

With $qW=5000$, $k_{min}=0.9$, the above probability will be smaller than
$5.77\cdot 10^{-12}$. Similarly, with $qW=10000,k_{min}=0.93$, the
probability will be smaller than $1.27\cdot 10^{-11}$. Since \sys is designed for large-scale training, it's reasonable to assume $k_{min}$ is close to 1. The argument is similar for $k_{max}$ (e.g. $k_{max}=1.01$).

However, if $qW$ is, for example, 500, to achieve a similar probability of $1.17\cdot 10^{-11}$, $k_{min}$ will be about 0.7. In this setting, it's not reasonable to assume $M'=qW$ as in Claim~\ref{sssec:aggclaim} any more. Instead one will need to assume a different bound, e.g. $M'<2qW$, and re-calculate the probability.
For example, if $M'=2qW$, in Claim~\ref{sssec:aggclaim}, the probability of an honest device checking a particular leaf node will still be the same. However, the probability of no honest device checking a particular node will be instead
$$(1-\frac{qs}{M'})^{(1-f)W}
<(1-\frac{s}{2W})^{(1-f)W}
\leq e^{-(1-f)s/2}$$
Verifiers are expected to verify more nodes in the summation tree to make Claim~\ref{sssec:aggclaim} true and thus ensure privacy.
\end{proof}

\emparagraph{Freshness.} 
Reusing the encryption key will not
affect the security proof in each round, but will lead to attacks across rounds, where information from previous rounds can be leaked.
The adversary may obtain the victim device's ciphertext $Enc(m)$ from a previous
round and use $k\cdot Enc(m)$, where $k$ is a large constant, to participate in
a later round. Since $k$ is large, $k\cdot Enc(m)$ will dominate the aggregated result. After decryption, the adversary will be able to learn approximately $m$. We define the 
 freshness as the guarantee that only fresh generated ciphertexts are included in the aggregation.

We need to prove that by asking each device to put the round number in the first
slot of ciphertext, generating corresponding ZK-proof, and in the add phase
verifying the ZK proof, the adversary will not be able to use ciphertexts from
previous rounds in the current round.

First, according to the knowledge soundness property of zkSNARK, which states that
it is not possible for a prover to construct a proof without knowing the
witness (e.g. secret inputs), an adversary can't construct a proof for a new
round~\cite{nitulescuzk}. This means the adversary can only insert a non-valid
proof into the summation tree.

Next, if the adversary inserts one ciphertext from a previous round with a non-valid proof into the
summation tree, with high probability, it will be caught, since as claimed
before, with high probability each leaf node will be checked by some honest
devices. The honest devices will
be able to detect this error.

\emparagraph{How \sys supports modified DP-FedAvg.} We've covered the integrity
and freshness of \sys's add phase/aggregation. Next we'll show the modified
DP-FedAvg will be executed faithfully. In order to show this, we claim the
following:

\begin{claim}
Data from honest generator devices will be included at most once in the aggregation; data from honest DP-noise committee members will be included exactly once.
\end{claim}
\begin{proof}
It's not difficult to see that both honest generators' and honest DP-noise
committee members' data will be included at most once, as defined in
Claim~\ref{sssec:aggclaim}. So for honest generator devices, the proof is done. 

As for honest DP-noise committee member, we just need to prove data from an
honest member will be included (at least once). Recall that for a honest
committee member, it is always online during a round, otherwise will be
considered as malicious. Also recall that in the add phase, after the aggregator
constructs the summation tree and the Merkle tree, the aggregator needs to send
a Merkle proof to the device that its data is included in the summation tree
(Step 7 in Figure~\ref{f:addphase}). An honest committee member can thus make
sure its data is included in the aggregation by verifying this Merkle proof.
\end{proof}

This claim captures the requirement for faithfully executing the modified
DP-FedAvg.

Firstly, data from honest generators (including online and offline) will be
included at most once, which is exactly what our new sampling method does
(Appendix~\S\ref{sssec:dp}): some honest generators might be filtered while
others not; those included in the aggregation will be added exactly once.

Secondly, data from honest DP-noise committee members (only online) will be
included exactly once, which ensures that enough Gaussian noise will be added.

Combining with the integrity and freshness of the underlying aggregation, the
modified DP-FedAvg protocol will be executed as it is in \sys.

\subsubsection{Decryption}\label{sssec:decrypt}
The last step in \sys's protocol is decryption.
In this section, we'll prove \sys's decryption protocol will not leak any
information, except the decrypted result.

Let's first review the BFV scheme~\cite{brakerski2012fully,fan2012somewhat}.
Let the  field for the coefficients of the ciphertext polynomials be $Q$, 
the polynomials themselves be from a polynomial ring $R_Q$,
the distribution $\phi$ for the coefficients of error polynomials 
be the required Gaussian distribution $\phi$
(standard deviation=3.2), and the secret key $sk$ be a polynomial of same degree $N$ as the ciphertext polynomials but
with coefficients from the ternary distribution (\{-1,0,1\}), which we
denote as $\psi$. Then, given a small constant $\gamma \ll 1$, the BFV scheme has the following procedures
\begin{itemize}
    \item $Keygen$: $s \leftarrow \psi, a\leftarrow R_Q, e\leftarrow \phi$. Compute $b=as+e$ and output $pk=(a,b)$ and $sk = s$
    \item $Enc(pk, m)$: $e_1\leftarrow \phi,e_2\leftarrow \phi,r\leftarrow \psi$, output $(c_1=ar+e_1,c_2=br+e_2+m/\gamma)$
    \item $Dec(c_1,c_2)$ : output $m=\lfloor (c_2-c_1s)\cdot \gamma\rceil$
\end{itemize}

Since $m$ will finally be known to the aggregator and revealing $m$ is safe,
without losing generality, assume
$(c_1,c_2)$ is the encryption of 0. Further, as defined in release phase (\S\ref{s:design:release}),
$$e_{small}=c_2-c_1\cdot s=re+e_2-se_1.$$
This small error must remain hidden during the decryption process; otherwise, it may reveal information about the 
secret key $s$ or the polynomial $r$. To hide $e_{small}$, \sys's scheme applies 
the smudging lemma~\cite{asharov2012multiparty}. This 
lemma states that to achieve $2^{-\lambda}$
statistical distance between $e_{small}$ and
$e_{small}+e_{smudging}$, $e_{smudging}$ just needs to be sampled from a uniform
distribution whose bound is $\lambda$ bits more than the upper bound of
$e_{small}$. Suppose the smudging distribution is $\phi'$, which is some uniform distribution whose bound is $\lambda$ bits more than the upper bound of $e_{small}$.

Recall that in the release phase, the decryption committee reveals $c_1\cdot s +
e_{smudging}$ to the aggregator, where $e_{smudging}\leftarrow \phi'$. So the adversary's view is
$$(c_2,c_1,c_1s+e_{smudging}),$$
which is equivalent to
$$(c_2,c_1,-c_2+c_1s+e_{smudging}).$$

If we can prove this view is indistinguishable from

$$(c_2,c_1,e_{smudging}')$$

where $e_{smudging}'$ is some freshly sampled error from the smudging distribution $\phi'$, then revealing $c_1\cdot s
+ e_{smudging}$ to the aggregator will not leak more information than telling
the aggregator a uniformly random number, since $(c_1,c_2)$ are already known to
the aggregator.

To prove the above claim, let's start with
$$(c_2,c_1,-c_2+c_1s+e_{smudging}).$$
Expanding $c_2$ and $c_1$, we  get that the above is same as
$$(asr+re+e_2,ar+e_1,-re-e_2+se_1+e_{smudging}).$$
With the smudging lemma, the above
is indistinguishable from
$$((as+e)r+e_2,ar+e_1,e_{smudging}').$$
Notice that $e_{smudging}'$ has nothing to do with the secret key $s$, and we can apply Ring-LWE assumption to convert $(as+e,a)$ back to $(b,a)$ as otherwise $(as+e,a)$ is not indistinguishable from $(b,a)$. The above is indistinguishable from
$$(br+e_2,ar+e_1,e_{smudging}')=(c_2,c_1,e'_{smudging}).$$ 
Since $e_{smudging}'$ doesn't depend on either the secret key or
honest devices' data, revealing partial
decryption result will not leak information about either the secret key or honest
devices' data.

\subsection{Details of the setup phase}
\label{s:appendix:design}
During the setup phase, \sys (i) forms the master committee, which then (ii) receives and validates inputs for the round, and (iii) generates keys for cryptographic primitives. We present the details for only the second piece here, as the first and the third are discussed in detail earlier (\S\ref{s:design}).

For the second piece, the master committee needs to check whether there is enough privacy budget to run a training task before launching it. To do this, the committee members need to calculate the new DP parameters before they launch the training task and check whether the new parameters are below some threshold. The details are as follows.

Recall that once the master committee is formed, each committee member receives the model parameters $\theta^t$ for the current round $t$, the user selection probability $q$, noise scale $z$, and clipping bound $S$ from the aggregator for this round of training (required by DP-FedAvg; Figure~\ref{f:dpfedavg}).

Each committee member locally computes new values of the DP parameters $\epsilon, \delta$ using the moment accounting algorithm $\mathcal{M}$ (line~\ref{l:dpfedavg:updateprivacybudget} in DP-FedAvg). This computation requires\changebars{ the DP parameters $\epsilon',\delta'$}{the internal state of $\mathcal{M}$} from round $t-1$ 
in addition to the inputs $z, q$. The committee member downloads the former from a public bulletin board, where it is signed by more than a threshold of honest members of the previous round's master committee. After getting the former DP parameters, the committee member calculates the new $\epsilon, \delta$. If the new values are below their recommended value, the committee member signs a certificate containing the parameters ($\theta^t$, $q$, $z$, $S$), new values of $\epsilon, \delta$, and keys for cryptographic primitives and publish it to the bulletin board to start this training task.

\end{document}